\documentclass[11pt,reqno]{article}

\usepackage[pdftex]{graphicx}
\usepackage{epstopdf}
\usepackage{graphicx}
\usepackage{float}
\usepackage{amsfonts}
\usepackage{amsmath}
\usepackage{amsthm}
\usepackage{amssymb}
\usepackage{mathrsfs}
\usepackage{amstext}
\usepackage{graphicx}
\usepackage{stmaryrd}
\usepackage{color}
\usepackage{url}
\usepackage{subfigure}
\usepackage{caption}
\font\msbm=msbm10

\numberwithin{equation}{section}

\theoremstyle{plain}
\newtheorem{theorem}{Theorem}[section]
\newtheorem{lemma}[theorem]{Lemma}

\newtheorem{definition}{Definition}[section]

\newtheorem{remark}[theorem]{Remark}
\def\mathbb#1{\hbox{\msbm{#1}}}

\newcommand{\N}{{\mathbb{N}}}
\newcommand{\R}{{\mathbb{R}}}
\newcommand{\Z}{{\mathbb{Z}}}
\newcommand{\C}{{\mathbb{C}}}

\renewcommand{\P}{{\mathbb{P}}}
\newcommand{\E}{{\mathbb{E}}}

\newcommand{\sgn}{\operatorname{sgn}}

\newcommand{\Qpar}{\mathcal{Q}}
\newcommand{\Rpar}{\mathcal{R}}

\newcommand{\Id}{\operatorname{Id}}

\newcommand{\beq}{\begin{eqnarray}}
\newcommand{\eeq}{\end{eqnarray}}
\newcommand{\beqn}{\begin{eqnarray*}}
\newcommand{\eeqn}{\end{eqnarray*}}

\newcommand{\ol}{\overline}

\DeclareMathOperator*{\diag}{diag}
\renewcommand{\diag}{\operatorname{diag}}
\DeclareMathOperator*{\sign}{sgn}
\DeclareMathOperator*{\argmin}{argmin}
\renewcommand{\Re}{\operatorname{Re}}

\usepackage[labelfont=bf]{caption}

\begin{document}
\title{Remote sensing via $\ell_1$-minimization}

\date{May 4, 2012; revised April 24, 2013}

\author{Max H\"ugel\thanks{Hausdorff Center for Mathematics and Institute 
for Numerical Simulation, University of Bonn, Bonn, Germany, {\tt max.huegel@hcm.uni-bonn.de}} \and
Holger Rauhut\thanks{Hausdorff Center for Mathematics and Institute for 
Numerical Simulation, University of Bonn, Bonn, Germany, {\tt rauhut@hcm.uni-bonn.de}}  \and
Thomas Strohmer\thanks{Department of Mathematics, University of 
California at Davis, Davis CA, {\tt strohmer@math.ucdavis.edu}}}
\maketitle
\begin{abstract} 
We consider the problem of detecting the locations of targets in 
the far field by sending probing signals from an antenna array and recording
the reflected echoes. Drawing on key concepts from the area of compressive sensing, we use an $\ell_1$-based regularization approach to solve
this, in general ill-posed, inverse scattering problem. As common in compressive sensing, we exploit randomness, which in this
context comes from choosing the antenna locations at random. With $n$ antennas we obtain $n^2$ measurements
of a vector $x \in \C^{N}$ representing the target locations and
reflectivities on a discretized grid. It is common to assume that the scene $x$ is sparse due to a limited
number of targets. Under a natural condition on the mesh size of the grid,
we show that an $s$-sparse scene can be recovered via $\ell_1$-minimization
with high probability if $n^2 \geq C s \log^2(N)$. The reconstruction is stable under noise and under passing from sparse to approximately sparse vectors.
Our theoretical findings are confirmed by numerical simulations.

\medskip
\noindent
{\bf AMS Subject Classification:} 65K05, 65C99, 65F22, 94A99, 90C25

\medskip
\noindent
{\bf Keywords:} Compressive sensing, sparsity, $\ell_1$-minimization,
inverse scattering, regularization
\end{abstract}

\section{Introduction}
\label{sec:intro}
Our aim is 
to detect the locations and reflectivities of remote targets (point 
scatterers) by sending probing signals from an antenna array and recording 
the reflected signals. This type of inverse scattering --- which has 
applications in radar, sonar, medical imaging, and microscopy --- is a rather 
challenging numerical problem. 
Typically the solution is not unique and instabilities in the presence 
of noise are a common issue. Standard techniques, such as matched field
processing~\cite{Tol93} or time reversal methods~\cite{BPT03,DMG05,JMO07} work well 
for the detection of very few, well separated targets. However, when the 
number of targets increases and/or some targets are adjacent to each other, 
these methods run into severe problems. Moreover, these methods have major
difficulties when the dynamic range between the reflectivities of the
targets is large.

In~\cite{fayastro09} a compressive sensing based approach to the inverse
scattering problem was proposed to overcome the ill-posedness of the
problem by utilizing the sparsity of the target scene. Here, sparsity
is meant in the sense that the targets typically occupy only a small fraction 
of the overall 
region of interest. As common in compressive sensing \cite{do06-2,carota06,fora11,ra09}, 
randomness is used and in this setup it is realized by
placing the antennas at random locations on a square.
It was proved in~\cite{fayastro09} that under certain
conditions it is possible to exactly recover the locations and reflectivities 
of the targets from noise-free measurements by solving an
$\ell_1$-regularized optimization problem, also known as {\em basis pursuit} 
in the compressive sensing literature. 

\begin{figure}[t]
\begin{center}
  \subfigure[]{\includegraphics[width=15cm,height=6cm]{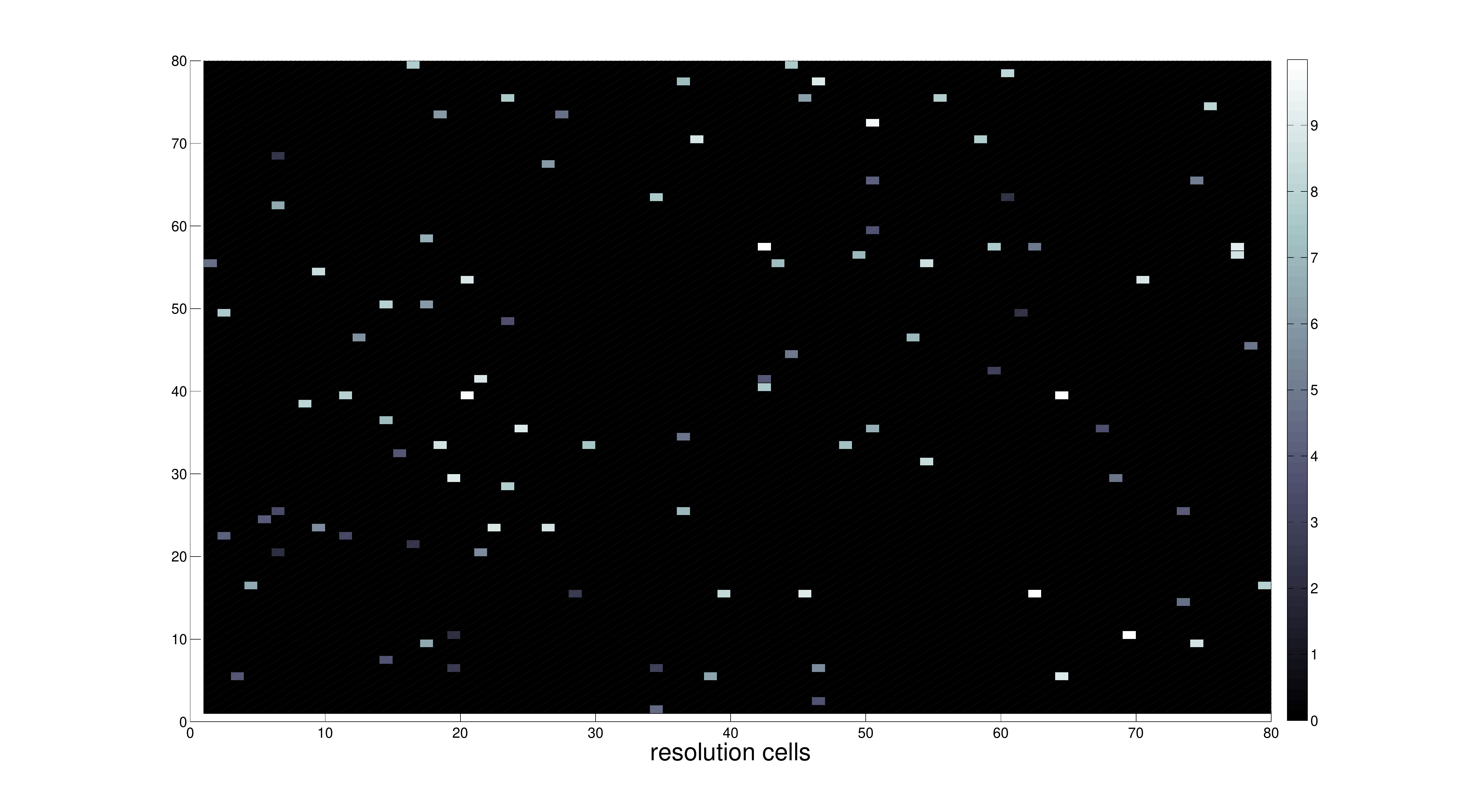}}
\subfigure[]{\includegraphics[width=15cm,height=6cm]{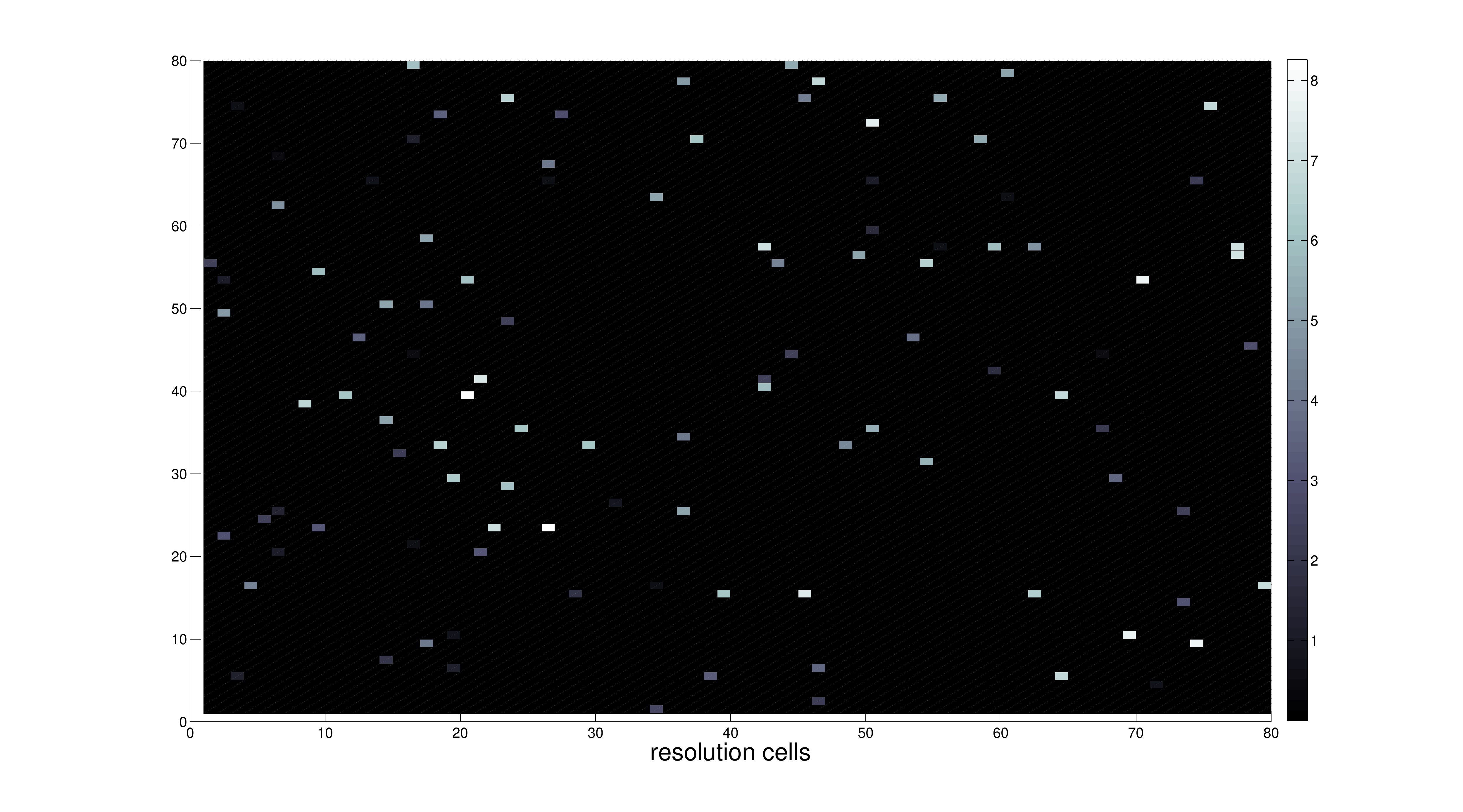}}
\end{center}
     \caption{(a) Scene with $100$ targets in $6400$ resolution cells (b)
Reconstruction from $900$ noisy measurements with SNR of $20$dB}
     \label{fig:examples:scene}
\end{figure}

While the framework in~\cite{fayastro09} can lead to significant improvements 
over traditional methods, it also has several limitations. For instance, the 
main theoretical result in that article requires the targets to be
randomly spaced, a condition
that is quite restrictive and does not match well with practical scenarios.
Also the conditions on the number of targets that can be recovered are
far from optimal. In this paper we will overcome most of these limitations, 
thus leading to a theoretical framework that is better adapted to 
practical applications. In particular, we also show that recovery is stable with respect to measurement noise
and under passing from sparse to approximately sparse scenes. Figure
\ref{fig:examples:scene} depicts the reconstruction of a sparse scene of
$100$ targets in $6400$ resolution cells with reflectivities in the dynamic range from $1$ to $8$ from $900$ noisy measurements, that is with $30$ antennas. Both the detection performance and the approximation of the true values of the reflectivities are very good.


What makes the inverse scattering problem
with antenna arrays challenging from a compressive sensing viewpoint is
that the associated sensing matrix is not a random matrix with independent 
rows or columns, but the matrix entries are random variables which are
coupled across rows and columns. This in turn means that standard proof 
techniques from the compressive sensing literature cannot be applied readily
and results developed for structured sensing matrices~\cite{ra09} 
are of limited use in our case. In fact, it is an open problem whether the by now 
classical and often used restricted isometry property holds for the random 
scattering matrix arising in our context. Instead we provide high probability 
recovery bounds for a fixed 
vector and a random choice of the scattering matrix (also referred to as nonuniform recovery guarantees).
We believe that some of the tools that we
develop in this paper will potentially be useful in other compressive sensing
scenarios, where the sensing matrix has coupled rows and columns.

Our paper is organized as follows. In Section~\ref{sec:random_sign} 
we describe the setup of the imaging problem and state our main results.
As preparation for proving our main theorems, we derive a general sparse 
recovery result in Section~\ref{sec:general} and condition number estimates
for certain random matrices in Section~\ref{sec:conditionnumbers}. In
Section~\ref{sec:nonuniform} we prove the recovery of sparse vectors for sensing matrices 
with dependent rows and columns which are associated with a class 
of bounded orthonormal systems. This type of matrices includes
the sensing matrix arising in the inverse scattering problem as a special
case. On the other hand this result assumes that the non-zero coefficients
of the signal to be recovered have random phases. In 
Section~\ref{sec:deterministic} we remove the assumption of random phases and show sparse recovery for the inverse scattering setup for signals with fixed deterministic phases. In Section~\ref{sec:numerics} we illustrate our
theoretical results by numerical simulations.

\subsection*{Acknowledgements}

M.H.\ and H.R.\ acknowledge support by the Hausdorff Center for Mathematics
and by the ERC Starting Grant SPALORA StG 258926. T.S.\ was supported by 
the National Science Foundation and DTRA under grant DTRA-DMS 1042939,  and by DARPA 
under grant N66001-11-1-4090. Parts of this manuscript have been written 
during a stay of H.R.\ at the Institute for Mathematics and Its Applications, 
University of Minnesota, Minneapolis. T.S.\ thanks Haichao Wang for useful
comments on an early version of this manuscript. The authors also wish to thank Axel Obermeier and the anonymous reviewers for helpful comments and corrections.

\section{Problem formulation and main results}
\label{sec:random_sign}

\begin{subsection}{Array imaging setup and problem formulation}\label{radar_setup}

Suppose an array of $n$ transducers is located in 
the square $[0,B]^2$, where $B>0$ is the array aperture.
The spatial part of a wave of wavelength $\lambda>0$ emitted from some point source $b\in[0,B]^2$ and recorded at another point $r\in\R^3$ is given by the 
Green's function $G$ of the Helmholtz equation,
\begin{equation}\label{eqnRadar1}
G(r,b):=\frac{\exp\left(\frac{2\pi i}{\lambda}\left\|r-b\right\|_2\right)}{4\pi\left\|r-b\right\|_2}.
\end{equation}
Here and in the following $\|\cdot\|_p$ refers to the usual $\ell_p$-norm.

Assume that we want to image the locations of targets which are at distance
$z_0>0$. For the analysis, we make the idealizing assumption that the
targets are on a discretized grid of meshsize $d_0>0$ in the target domain $TD:=[-L,L]^2\times\left\{z_0\right\}$, 
where $L>0$ determines the size of the target domain. To be more precise, let us assume that each target occupies 
one of the points $\left(r_j\right)_{j\in [N]}\subset TD$, where $[N] := \{1,\hdots,N\}$ with $N=\left\lfloor 2L/d_0\right\rfloor^2$ 
and each $r_j$ is of the form $r_j=\left(-L+kd_0,-L+\ell d_0,z_0\right)^T$ for some $(k,\ell) \in [\sqrt{N}]^2$. 
See also Figure \ref{figRadar_setup} 
for a visualization of this setup.

\begin{figure}[t]
\centering
   \includegraphics[width=15cm,height=6cm]{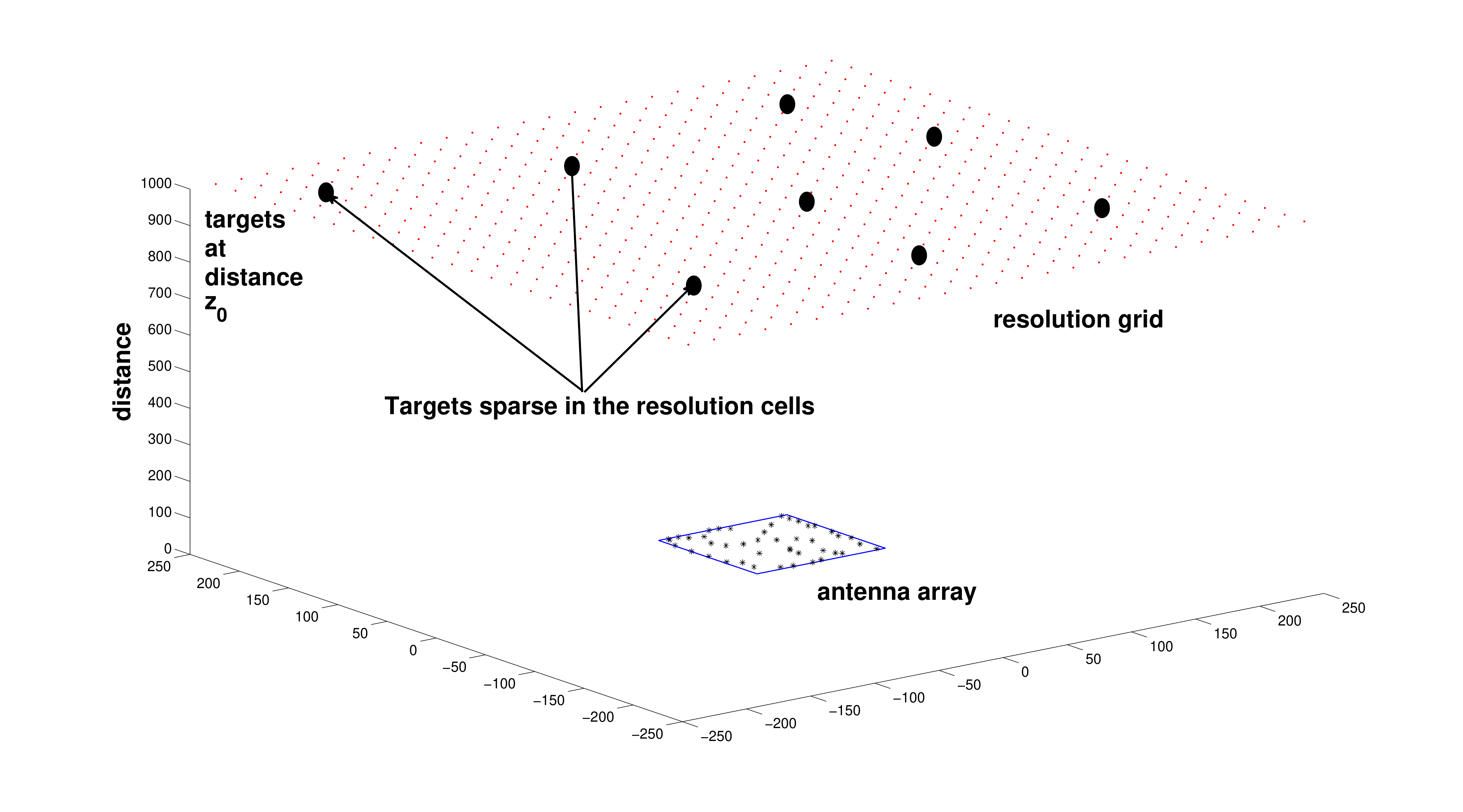}
   \caption{The targets at distance $z_0$ distributed sparsely in the
target domain}
\label{figRadar_setup}
\end{figure}

In order to be able to analyze the arising sensing mechanism, we approximate the Green's function from (\ref{eqnRadar1}) in an adequate way. 
To this end, we assume to be in the far field region, that is, the distance $z_0$ from antenna to target satisfies $z_0\gg B+L$. 
Writing $r=\left(x,y,z_0\right)^T$ and $b=\left(\xi,\eta,0\right)^T$, 
the truncated Taylor expansion for $\left\|r-b\right\|_2$ around $r_0:=\left(\xi,\eta,z_0\right)$ is given by
\begin{equation}\label{eqnRadar2}
\left\|r-b\right\|_2\approx z_0+\frac{\left\|(x,y)-(\xi,\eta)\right\|_2^2}{2z_0}.
\end{equation}
Under the far field assumption
we obtain then that
\begin{equation}\label{eqnRadar3}
G(r,b)\approx\widetilde{G}(r,b):=\exp\left(\frac{2\pi iz_0}{\lambda}\right)\frac{\exp\left(\frac{\pi i}{\lambda z_0}\left\|(x,y)-(\xi,\eta)\right\|_2^2\right)}{4\pi z_0}.
\end{equation}
If we choose the meshsize $d_0$ such that the
crucial \emph{aperture condition}~\cite{fayastro09}
\begin{equation}\label{eqnRadar4}
\rho:=\frac{d_0 B}{\lambda z_0}\in\N
\end{equation}
is fulfilled, then the normalized system of functions 
\[
\widehat{G}(b,r_\ell):= 4\pi z_0\widetilde{G}\left(b,r_\ell\right), \quad b \in [0,B], \ell  \in [N],
\]
satisfies the convenient orthonormality relation
\begin{align}
&\frac{1}{B^2}\int_{[0,B]^2} \widehat{G}(b,r_m)\overline{\widehat{G}(b,r_\ell)}db=\frac{\exp\left(\frac{\pi i}{\lambda z_0}\left(\left\|(x_m,y_m)\right\|_2^2-\left\|(x_\ell,y_\ell)\right\|_2^2\right)\right)}{B^2}\nonumber\\\
&\times\int_{[0,B]}\int_{[0,B]}\exp\left(-\frac{2\pi i}{\lambda z_0}(x_m-x_\ell)\xi\right)\exp\left(-\frac{2\pi i}{\lambda z_0}(y_m-y_\ell)\eta\right)d\xi d\eta\nonumber\\
&=\delta_{\ell m}.\label{eqnRadar5}
\end{align}
It is for this relation to hold that we make the approximation (\ref{eqnRadar3}).

Let us now describe the scattering matrix. Assume we have a vector $(x_j)_{j\in[N]}\in\C^N$ of reflectivities on the resolution grid. 
We sample $n$ antenna positions $b_1,\hdots,b_n\in[0,B]^2$ independently at random according to the uniform distribution on $[0,B]^2$. 
If antenna element $b_j\in[0,B]^2$ transmits and $b_k\in[0,B]^2$ receives, then we model the echo $y_{jk}$ as
\begin{equation}\label{eqnRadar6}
y_{jk}=\sum_{\ell=1}^N\widehat{G}(b_j,r_\ell)\widehat{G}(r_\ell,b_k)x_\ell, \quad (j,k) \in [n]^2.
\end{equation}
This is called the \emph{Born approximation}~\cite{BW99a}. It amounts to discarding multipath scattering effects. 
So if the transmit-receive mode is that one antenna element transmits at a time and the whole aperture receives the echo, 
the appropriately scaled sensing matrix $A\in\C^{n^2\times N}$ is given entrywise by
\begin{equation}\label{eqnRadar7}
A_{(j,k),\ell} := 
\widehat{G}(b_j,r_\ell)\widehat{G}(r_\ell,b_k), \quad
(j,k) \in[n]^2, \ell \in [N].
\end{equation}
Then $y = A x$ by \eqref{eqnRadar6}.
Due to the randomness in the $b_k$, $k\in [n]$, the matrix $A$ is a (structured) 
random matrix with coupled rows and columns. 

In many scenarios the number of targets is small compared to the grid size. 
This naturally leads to sparsity in the vector $x \in \C^{N}$ of reflectivities, $\|x\|_0 := \#\{\ell: x_\ell \neq 0\} \leq s$, where
$s \ll N$. Compressive sensing suggests that in such a scenario, 
we can recover $x$ from undersampled measurements $y = A x \in \C^{n^2}$
when $n^2 \ll N$. We note that $A$ contains only $n(n+1)/2$ different rows
due to the symmetries in the sensing setup.
Our goal is determine a good bound on the required minimal number of antennas $n$ in order to 
ensure recovery of an $s$-sparse scene. A small number of antennas has clear advantages such as low costs of imaging hardware.

\end{subsection}

\begin{subsection}{Compressive sensing}\label{section_cs}

We briefly describe the basics of compressive sensing in order to place our results outlined below into context.
Given measurements $y = A x \in \C^m$ of a sparse vector $x \in \C^N$, where $A \in \C^{m \times N}$ is the so-called measurement matrix, 
we would like to reconstruct $x$ in the underdetermined case that $m \ll N$ by taking into consideration the sparsity.

The na{\"i}ve approach of $\ell_0$-minimization
\begin{equation}\label{eq:l0min}
\min_{z \in \C^{N}} \|z\|_0 \quad \mbox{ subject to } A z = y
\end{equation}
is NP-hard \cite{na95}. Hence several tractable alternatives were proposed including $\ell_1$-minimization, also called basis pursuit \cite{chdosa98,do06-2,carota06},
\begin{equation}\label{eq:l1min}
\min_{z \in \C^{N}} \|z\|_1 \quad \mbox{ subject to } Az = y.
\end{equation}
This can be seen as a convex relaxation of \eqref{eq:l0min} and can be solved via efficient convex optimization methods \cite{bova04,chpo11}.
It is by now well-understood that $\ell_1$-minimization can recover sparse vectors under appropriate conditions. Remarkably, random matrices
provide (near-)optimal measurement matrices in this context and good deterministic constructions are lacking to date, see \cite{ra09,fora11} for a
discussion. For instance, an $m \times N$ Gaussian random matrix $A$ ensures exact (and stable) recovery of all $s$-sparse
vectors $x$ from $y = Ax$ using $\ell_1$-minimization (and other types of algorithms) with high probability provided
\begin{equation}\label{rec:cond:Gaussian}
m \geq C s \log(N/s),
\end{equation}
where $C>0$ is a universal constants. This bound is optimal
\cite{do06-2,foparaul10}. It is crucial that $m$ is allowed to scale
linearly in $s$. 
The $\log$-factor cannot be removed. Recovery is stable under passing to
approximately sparse vectors and under adding noise to the 
measurements. In the latter case, one may rather work with the noise-constrained $\ell_1$-minimization problem
\begin{equation}\label{eqnN2}
\min_{z \in \C^{N}} \|z\|_1 \quad \mbox{ subject to } \|Az - y\|_2 \leq \eta.
\end{equation}
Random partial Fourier matrices \cite{carota06,cata06,ruve08,ra05-7,ra09} (that is, random row-submatrices of the discrete Fourier matrix) 
and other types of structured random matrices \cite{ra09,pfra10} also provide $s$-sparse recovery
under similar conditions as in \eqref{rec:cond:Gaussian} (with additional $\log$-factors). 

Some of the mentioned recovery results are derived using the restricted isometry property (RIP) \cite{cata06,carota06-1}.
This leads to uniform guarantees in the sense that once the matrix is
selected, then with high probability {\it every} $s$-sparse vector can be 
recovered from $y = Ax$. The RIP, however, is a rather strong condition which is sometimes hard to verify. In particular, it is open
to verify it for our random matrix in \eqref{eqnRadar7}. Instead, we may work with weaker conditions, 
which ensure nonuniform recovery in the sense that a fixed $s$-sparse vector is recovered with high probability using a random draw of the matrix.
Our result below for the structured random matrix in \eqref{eqnRadar7} is based on the extension of certain general recovery conditions for $\ell_1$-minimization \cite{fu04,tr05-1,capl11-1} to stable recovery
using a so-called dual certificate, see Section \ref{sec:general}.

\end{subsection}

\begin{subsection}{Main results}\label{section_submatrices}


We define the 
error of best $s$-term approximation in the $\ell_1$-norm by 
\begin{equation*}
\sigma_s(x)_1:=\inf_{\|z\|_0 \leq s} \|x-z\|_1. 
\end{equation*}
Furthermore, we will assume throughout that the \emph{aperture condition}
\begin{equation}\label{eqnApCond}
\rho:=\frac{d_0B}{\lambda z_0}\in\N
\end{equation}
holds, which can be accomplished by an appropriate choice of the meshsize
$d_0$. The further notation is the one used in Section \ref{radar_setup}. 
We will refer to the matrix $A \in \C^{n^2 \times N}$ 
in \eqref{eqnRadar7} with the antenna positions $b_1,\hdots,b_n$ selected independently and uniformly at random from $[0,B]^2$
as the {\it random scattering matrix}. Note that the aperture condition \eqref{eqnApCond} implies that $\E A^*A = n^2\Id$ by a similar
computation as in \eqref{eqnRadar5}, that is, in expectation
the matrix $A^*A$ behaves nicely, which will be crucial in the proof. Let us now state our nonuniform recovery result.

\begin{theorem}\label{Thm_det_patterns}
Let $x\in\C^N$ 
and $A\in\C^{n^2\times N}$ be a draw of the random scattering matrix.
Let $s\in \N$ be some sparsity level. 
Suppose we are given noisy measurements $y=Ax+e \in \C^{n^2}$ with $\left\|e\right\|_2\leq\eta n$. 
If, for $\varepsilon > 0$,
\begin{equation}\label{eqnN1}
n^2\geq Cs\log^2\left(\frac{cN}{\varepsilon}\right) 
\end{equation}
with universal constants $C,c>0$,
then with probability at least $1-\varepsilon$, the solution $\widehat{x}\in\C^N$ to the noise-constrained $\ell_1$-minimization problem 
\begin{equation}\label{eq:BPDN}
\min_{z \in \C^{N}} \|z\|_1 \quad \mbox{ subject to } \|Az - y\|_2 \leq \eta n.
\end{equation}
satisfies
\begin{equation}\label{eqnN3}
\left\|x-\widehat{x}\right\|_2\leq C_1 \sqrt{s} \eta+C_2 \sigma_s(x)_1.
\end{equation}
The constants satisfy $C\leq\left(800 e^{3/4}\right)^2 \approx 2.87\cdot 10^6$, $c\leq 6$, $C_1\leq4(1+\sqrt{2})+8\sqrt{3}\approx 23.513$, $C_2\leq4(1+\sqrt{6})\approx 13.798$.
\end{theorem}
\begin{remark} 
\begin{itemize}
\item[(a)] The constants appearing in Theorem \ref{Thm_det_patterns} are quite large and reflect a worst case analysis. No attempt has been made to optimize the above bounds. In practice, much better bounds can be expected, see also the numerical results below.
\item[(b)] The scaling of the noise level, $\left\|e\right\|_2\leq\eta n$ is natural because $e \in \C^{n^2}$.
Indeed, if we have a componentwise bound $|e_j| = |(A x)_j - y_j| \leq \eta$ for all $j \in [n]^2$ then it is satisfied. 
\item[(c)] The error bound \eqref{eqnN3} is slightly worse than the one we would get under the RIP. In fact, if $A$ has
the RIP then the associated error bound improves the right hand side of \eqref{eqnN3} by a factor of $s^{-1/2}$
\cite{carota06-1}. Unfortunately, it is so far unknown whether the random scattering matrix $A$ obeys the RIP under 
a similar condition as \eqref{eqnN1}, so that the error bound \eqref{eqnN3} is the best one can presently achieve.
\item[(d)] If $x$ is $s$-sparse, $\sigma_s(x)_1=0$, and if there is no noise, $\eta = 0$, then \eqref{eqnN3} implies 
exact reconstruction, $\widehat{x} = x$, by equality-constrained $\ell_1$-minimization \eqref{eq:l1min}.
\item[(e)] We can specialize the error bound in the previous theorem for the case of 
Gaussian noise. To this end, assume that
the components of $e \in \C^{n^2}$ are i.i.d.\ complex Gaussians with variance
$\eta^2$, where the real and imaginary part of a complex Gaussian are independent 
real Gaussians with variance $\eta^2/2$. A standard calculation shows that the noise satisfies $\left\|e\right\|_2\leq\eta n\log(1/\varepsilon)$ with probability at least $1-\varepsilon$. Assuming that $e$ is independent of the matrix $A$, it follows that the solution $\hat{x}$ of noise-constrained $\ell_1$-minimization with bound $\left\|Az-y\right\|_2\leq\eta n\log(1/\varepsilon)$ satisfies
\begin{equation}
\label{eqnN3a}
\left\|\hat{x}-x\right\|_2\leq C_1\eta\sqrt{s}\log(1/\varepsilon)+C_2\sigma_s(x)_1
\end{equation} 
with probability at least $(1-\varepsilon)^2$. The constants $C_1$, $C_2$ satisfy the bounds of Theorem \ref{Thm_det_patterns}.
\end{itemize}
\end{remark}


Theorem \ref{Thm_det_patterns} holds for a fixed, deterministic
$x\in\C^N$. We define the \emph{sign} of a number $a\in\C$ as
\begin{equation*}
\sign(a) =
\begin{cases}
\frac{a}{\left|a\right|} & \text{if } a \not= 0,\\
0 & \text{if } a = 0.
\end{cases}
\end{equation*} 
For a vector $x\in\C^N$ we denote by $\sign(x):=\left(\sign(x_j)\right)_{j\in[N]}$ the sign pattern of $x$. On the way to the proof of Theorem \ref{Thm_det_patterns}, we will provide the easier result stated next 
for the case when the sign pattern of $x$ restricted to its support set
$T\subset[N]$, $\sign(x)_T=\left(\sign(x_j)\right)_{j\in T}$, forms a Rademacher or a Steinhaus
sequence. 
The latter amounts to assuming that the phases of the reflectivities
are iid uniformly distributed on $[0,2\pi]$, which is a common assumption in
array imaging and radar signal processing.
Theorem  \ref{Thm_rand_patterns} below actually establishes sparse recovery in a more general
setting than the inverse scattering problem. It is not only applicable to
the radar-type sensing matrices analyzed above, but to more general sensing
matrices whose rows and columns are not independent, and whose entries
are associated with a certain class of orthonormal systems. 
Its statement requires the notion of bounded orthonormal systems \cite{ra09}.

\begin{definition}\label{defi_BOS}
Let $D\subset\R^d$ be a measurable set and $\nu$ a probability measure on $D$. 
A system of functions $\left\{\Phi_k:D\rightarrow \C\right\}_{k\in[N]}$ is called a \emph{bounded orthonormal system} (BOS)
with respect to $(D,\nu)$ if
\begin{equation*}
\int_D\Phi_k(t)\overline{\Phi_\ell(t)}\nu(dt)=\delta_{k\ell}
\end{equation*}
and if the functions are uniformly bounded by a constant $K\geq 1$,
\begin{equation*}
\max_{k\in[N]}\left\|\Phi_k\right\|_\infty\leq K.
\end{equation*}
\end{definition} 

Let now $\left\{\Phi_\ell \right\}_{\ell \in[N]}$ be a BOS on $(D,\nu)$ with bounding constant $K=1$ and with the 
property that $\left\{\Phi_\ell^2\right\}_{\ell \in[N]}$ is also a BOS on $(D,\nu)$. Note that due to the orthogonality relation, we
then necessarily have $\left|\Phi_\ell(t)\right|=1$ for all $t\in D$. The functions $\Phi_\ell(t) = \widehat{G}(r_\ell,t)$, $t \in [0,B]^2$ fall
into this setup when the aperture condition \eqref{eqnApCond} is satisfied, see also \eqref{eqnRadar5}. Another example is provided 
by the Fourier system $\{\Phi_\ell\}_{\ell \in \Z}$, where
$\Phi_\ell(t) = e^{2\pi i \ell t}$, $\ell \in \Z$, $t \in [0,1]$.
For $b_1,b_2\in D$, set
\begin{equation*}
v(b_1,b_2):=\left(\overline{\Phi_\ell(b_1)}\,\overline{\Phi_\ell(b_2)}\right)_{\ell\in[N]}\in\C^N.
\end{equation*}
Sample now $n$ elements $b_1,\hdots,b_n$ independently at random according to $\nu$ from $D$. Define the sampling matrix $A$ via
\begin{equation}\label{eqn_BOS_sample}
A:=\left(v(b_j,b_k)^*\right)_{j,k\in[n]}\in\C^{n^2\times N},
\end{equation}
so that $A$ is the matrix with rows $v(b_j,b_k)^*$, $(j,k) \in [n]^2$. Note that 
with the system $\Phi_\ell(b) = \widehat{G}(r_\ell,b)$ we recover the random 
scattering matrix \eqref{eqnRadar7} in this way.

Now we can state our main result for random sign patterns. We recall that the entries of a (random) Rademacher vector $\mathbf{\epsilon}$
are independent random variables that take the values $\pm 1$ with equal probability. Similarly, a Steinhaus vector is a random vector
where all entries are independent and uniformly distributed on the complex torus $\{z \in \C: |z| = 1 \}$.
\begin{theorem}\label{Thm_rand_patterns}
Let $A\in\C^{n^2\times N}$ be a draw of the random sampling matrix from (\ref{eqn_BOS_sample}). Let $x\in\C^N$ and $T\subset[N]$ be the index
set corresponding to its $s$ largest absolute entries.
Assume that the sign vector $\sign(x)_T$ of $x$ restricted to $T$ 
forms a Rademacher or a Steinhaus sequence. Suppose we take noisy measurements 
$y=Ax+e\in\C^{n^2}$ with $\left\|e\right\|_2\leq\eta n$. If 
\begin{equation}\label{eqn_rand_cond}
n^2\geq C s\log\left(\frac{c_1(N-s)}{\varepsilon}\right)\log^2\left(c_2(N-s)^2s/\varepsilon\right), 
\end{equation} 
then with probability at least $1-\varepsilon$, the solution $\widehat{x}\in\C^N$ to
noise-constrained $\ell_1$-minimization \eqref{eq:BPDN} satisfies
\begin{equation}\label{eqn_rand_rec_err}
\left\|\widehat{x}-x\right\|_2\leq C_1\sqrt{s}\eta+C_2\sigma_s(x)_1.
\end{equation}
The constants satisfy $C\leq 1024$, $c_1\leq 8$, $c_2\leq 576$, $C_1\leq4(1+\sqrt{2})+8\sqrt{3}\approx 23.513$, $C_2\leq4(1+\sqrt{6})\approx 13.798$.
\end{theorem}
\begin{remark}
Whereas the bounds on the constants in Theorem \ref{Thm_det_patterns} are quite large, and certainly improvable, in the case of random sign patterns, the number of antennas required must satisfy
\begin{equation*}
n\geq 32\sqrt{s}\log^{3/2}\left(cN/\varepsilon\right),
\end{equation*}
which is a reasonable bound, see also the improvement in Remark \ref{remark_ic} (b).
\end{remark}
\end{subsection}

\section{Stable sparse recovery via $\ell_1$-minimization}
\label{sec:general}

In this section we establish a general result for the recovery of an individual vector $x\in\C^N$ from noisy measurements 
$y=Ax+e \in\C^m$ with $A\in\C^{m\times N}$. It uses a dual vector in the spirit of \cite{fu04,tr05-1} 
and extends these results to the noisy and non-sparse case.
The proof is inspired by \cite{capl11-1} for recovery based on the weak RIP. However, since we actually 
do not assume the weak RIP, the error bound in \eqref{eqnN9} below is slightly worse by a factor of $\sqrt{s}$ 
than the one in \cite[Section 4]{capl11-1}. In the noiseless and exact sparse case the theorem below implies exact recovery similar to \cite{fu04,tr05-1}. 

For a set $T \subset [N]$ and a matrix $A \in \C^{m \times N}$
with columns $a_j \in \C^m$, $j \in [N]$, we denote by $A_T = (a_j)_{j \in T}\in\C^{m\times\left|T\right|}$ the column-submatrix of $A$ with columns indexed by $T$ and by $T^c:=[N]\setminus T$ the complement of $T$ in $[N]$.
Similarly, we denote by $x_T\in\C^{|T|}$ the vector $x \in \C^N$ restricted to its entries in $T$. The operator norm of a matrix $B$ on $\ell_2$ 
is denoted by $\|B\|_{2 \to 2}$. 

\begin{theorem}\label{Thm_noisy_cond} Let $x\in\C^N$ and $A\in\C^{m\times N}$ with $\ell_2$-normalized columns, 
$\left\|a_j\right\|_2=1$, $j\in [N]$.
For $s\geq 1$, let $T\subset[N]$ be the set of indices of the $s$ largest absolute entries of $x$. 
Assume that $A_T$ is well-conditioned in the sense that
\begin{equation}\label{eqnN7}
\left\|A_{T}^*A_{T}-\Id\right\|_{2\to 2}\leq\frac{1}{2}
\end{equation} 
and that there exists a \emph{dual certificate} $u=A^*v\in\C^N$ with $v\in\C^m$ such that 
\begin{align}
u_T&=\sign(x)_T,\label{eqnN4}\\
\left\|u_{T^c}\right\|_\infty&\leq\frac{1}{2},\label{eqnN5}\\
\left\|v\right\|_2&\leq \sqrt{2s}.\label{eqnN6}
\end{align}
Suppose we are given noisy measurements $y=Ax+e\in\C^m$ 
with $\left\|e\right\|_2\leq\eta$. Then the solution $\widehat{x}\in\C^N$ to 
noise-constrained $\ell_1$-minimization \eqref{eqnN2} satisfies
\begin{equation}\label{eqnN9}
\left\|x-\widehat{x}\right\|_2\leq C_1\sqrt{s}\eta+C_2\sigma_s(x)_1.
\end{equation}
The constants satisfy $C_1\leq4(1+\sqrt{2})+8\sqrt{3}\approx 23.513$, $C_2\leq4(1+\sqrt{6})\approx 13.798$.
\end{theorem}
\begin{remark} The constants appearing in the conditions above are rather arbitrary and chosen for convenience.
\end{remark}
\begin{proof} 
Write $\widehat{x}=x+h$. Due to (\ref{eqnN2}) and the assumption on the noise level, $\|e\|_2 \leq \eta$, we have
\begin{equation}\label{eqnN10}
\left\|Ah\right\|_2 = \|A x - y - (A\widehat{x} - y)\| \leq \|Ax-y\|_2 + \|A \widehat{x} - y\|_2 \leq 2\eta.
\end{equation}
Since $x$ is feasible for the optimization program \eqref{eqnN2} we obtain
\begin{align*}
\left\|x\right\|_1\geq\left\|\widehat{x}\right\|_1&=\left\|(x+h)_T\right\|_1+\left\|(x+h)_{T^c}\right\|_1\\
&\geq\Re\left(\langle(x+h)_T,\sgn(x)_T\rangle\right)+\left\|h_{T^c}\right\|_1-\left\|x_{T^c}\right\|_1\\
&=\left\|x\right\|_1+\Re\left(\langle h_T,\sgn(x)_T\rangle\right)+\left\|h_{T^c}\right\|_1-2\left\|x_{T^c}\right\|_1,
\end{align*}
where we applied H\"older's and the triangle inequality in the second line.
Rearranging the above yields
\begin{equation}\label{eqnN11}
\left\|h_{T^c}\right\|_1\leq\left|\Re\left(\langle h_T,\sgn(x)_T\rangle\right)\right|+2\left\|x_{T^c}\right\|_1.
\end{equation}
Let $u=A^*v$ be the dual certificate. Then, using the Cauchy-Schwarz and H\"older's inequality
\begin{align*}
\left|\Re\left(\langle h_T,\sgn(x)_T\rangle\right)\right|&=\left|\Re\left(\langle h_T,\left(A^*v\right)_T\rangle\right)\right|
\leq\left|\langle h,A^*v\rangle\right|+\left|\langle h_{T^c},u_{T^c}\rangle\right|\\
&\leq\left\|Ah\right\|_2\left\|v\right\|_2+\left\|h_{T^c}\right\|_1\left\|u_{T^c}\right\|_\infty
\leq 2\sqrt{2s}\eta+\frac{1}{2}\left\|h_{T^c}\right\|_1,
\end{align*}
where we used (\ref{eqnN5}) and (\ref{eqnN6}) in the last line. Plugging into (\ref{eqnN11}) yields
\begin{equation}\label{eqnN12}
\left\|h_{T^c}\right\|_1\leq 4\sqrt{2s}\,\eta+4\left\|x_{T^c}\right\|_1.
\end{equation}
Due to (\ref{eqnN7}), we have 
\begin{align}
\frac{1}{2}\left\|h_T\right\|_2^2&\leq\left\|A_Th_T\right\|_2^2 =\langle A_Th_T,Ah\rangle-\langle A_Th_T,A_{T^c}h_{T^c}\rangle.\label{eqnN13}
\end{align}
Using H\"older's inequality, the normalization of the columns of $A$ and (\ref{eqnN10}), we obtain
\begin{align*}
\left|\langle A_Th_T,Ah\rangle\right|&\leq\left\|h_T\right\|_1\left\|A_T^*Ah\right\|_\infty \leq2\sqrt{s}\eta\left\|h_T\right\|_2.
\end{align*}
The triangle inequality and the Cauchy Schwarz inequality give, by noting that \eqref{eqnN7} implies $\left\|A_T\right\|_{2\to 2}\leq\sqrt{\frac{3}{2}}$,
\begin{align*}
\left|\langle A_Th_T,A_{T^c}h_{T^c} \rangle\right|&\leq \sum_{j \in T^c} |h_j| |\langle A_T h_T, a_j\rangle|
\leq \sum_{j \in T^c} |h_j| \|A_T h_T\|_2 \| a_j\|_2 \leq \sqrt{\frac{3}{2}} \|h_T\|_2 \|h_{T^c}\|_1. 
\end{align*}
Inserting into (\ref{eqnN13}) we obtain
\begin{equation}\label{eqnN14}
\left\|h_T\right\|_2\leq4\sqrt{s}\eta+\sqrt{6}\left\|h_{T^c}\right\|_1.
\end{equation}
Combining (\ref{eqnN12}) and (\ref{eqnN14}) we arrive at
\begin{align*}
\left\|h\right\|_2&\leq\left\|h_T\right\|_2+\left\|h_{T^c}\right\|_1\\
&\leq (4(1+\sqrt{2})+8\sqrt{3})\sqrt{s}\eta+4(1+\sqrt{6})\left\|x_{T^c}\right\|_1.
\end{align*}
Due to the choice of $T$ we have $\left\|x_{T^c}\right\|_1=\sigma_s(x)_1$. This completes the proof.
\end{proof}

\begin{section}{Conditioning of submatrices}
\label{sec:conditionnumbers}

Theorem \ref{Thm_noisy_cond} requires to find a dual certificate $u=A^*v$ with $u_T = \sign(x)_T$, where $A$ is the random
scattering matrix introduced in Section \ref{radar_setup} and $T\subset[N]$ is some support set.
Condition (\ref{eqnN7}) in Theorem \ref{Thm_noisy_cond} suggests to investigate the conditioning of $A_T$. 
Recall that 
\begin{equation*}
v(b_j,b_k)=\left(\overline{\Phi_\ell(b_j)}\, \overline{\Phi_\ell(b_k)}\right)_{\ell \in [N]}\in\C^N,
\end{equation*} 
where $\{\Phi_\ell\}$ is a bounded orthonormal system with constant $K=1$ such that $\{\Phi_\ell^2\}$ is also a bounded orthonormal
system. The rows of the random scattering matrix $A \in \C^{n^2 \times N}$ are the vectors $v(b_j,b_k)^* \in \C^{1\times N}$, $(j,k) \in [n]^2$, 
where the $b_1,\hdots, b_n$ are selected independently at random according to the orthonormalization measure $\nu$, see \eqref{eqn_BOS_sample} and 
Definition \ref{defi_BOS}. The scattering matrix $A$ in \eqref{eqnRadar7} is a special case of this setup.

We aim at a probabilistic estimate of the largest and smallest singular value of 
$\frac{1}{n} A_T \in \C^{n^2 \times s}$, i.e., the operator norm
\begin{equation}\label{eqn1}
\left\|\frac{1}{n^2} A_T^* A_T - \Id \right\|_{2 \to 2} = \left\|\frac{1}{n^2} \sum_{j,k=1}^n v(b_j,b_k)_T v(b_j,b_k)_T^* - \Id \right\|_{2 \to 2}.
\end{equation}
The central result of this section stated next provides an estimate of the tail of this quantity.


\begin{theorem}\label{central:thm}
Let $A\in\C^{n^2\times N}$ be the random matrix described above and let $T\subset[N]$ be a (fixed) subset of cardinality $\left|T\right|=s$. 
If, for $\delta,\varepsilon>0$, 
\begin{equation}\label{eqnA17}
n^2\geq 1024\delta^{-2}s\log^2\left(\frac{576 s^3}{\varepsilon}\right)
\end{equation}
then 
\begin{equation}\label{eqnA18}
\P\left(\left\|\frac{1}{n^2}A_{T}^*A_{T}-\Id\right\|_{2\to 2}\geq \delta\right)\leq\varepsilon.
\end{equation}
\end{theorem}
The proof will be given after some auxiliary results are presented.
\subsection{Auxiliary results}


The fact that the rows of $A$ are not independent makes the analysis difficult at first sight.
%
In order to increase the amount of independence, we will use a version of the tail decoupling inequality in 
Theorem~3.4.1 of \cite{gide99} .
For convenience, we provide a short proof, which essentially repeats the one in \cite{decoup_inequ} 
in our slightly more general setup.
In this way, we also obtain better constants than by tracing the ones in the proof of \cite[Theorem 3.4.1]{gide99}. 

\begin{theorem}\label{Decoup_thm}
Let $\left(X_i\right)_{i\in[n]}$, $n\geq 2$, be independent random variables with values in a measurable space $\Omega$. 
Let $h:\Omega\times\Omega\rightarrow B$ be a measurable map with values in a separable Banach space 
$B$ with norm $\left\|\cdot\right\|$. Then there exists a subset $S\subset [n]$ such that
\begin{align}
\P\left(\left\|\sum_{i\not= j}h(X_i,X_j)\right\|>t\right)\leq 36\,&\P\left(4\left\|\sum_{i\in S,j\in S^c}h(X_i,X_j)\right\|>t\right)\vee\nonumber\\&36\,\P\left(4\left\|\sum_{i\in S^c,j\in S}h(X_i,X_j)\right\|>t\right),
\end{align}
where for $a,b\in\R$ we denote $a\vee b:=\max\left\{a,b\right\}$.
\end{theorem}

The proof of Theorem \ref{Decoup_thm} employs Corollary 3.3.8 from \cite{gide99}. 

\begin{lemma}\label{lemma_lin_form}
Let $(B,\|\cdot\|)$ be a separable Banach space 
and let $Y$ be a $B$-valued random vector such that for each $\xi\in B^*$, the dual space of $B$, the map $\xi(Y)$ is measurable, 
centered and square integrable. Then, for every $x\in B$,
\begin{equation}\label{eqnAux2}
\mathbb{P}\left(\left\|x+Y\right\|\geq\left\|x\right\|\right)>\frac{1}{4}\inf_{\xi\in B^*}\left(\frac{\mathbb{E}\left[\left|\xi(Y)\right|\right]}{\left(\mathbb{E}\left[\left|\xi(Y)\right|^2\right]\right)^{1/2}}\right)^2.
\end{equation}
\end{lemma}
{\em Proof of Theorem \ref{Decoup_thm}}. Set $\mathcal{D}:=\left(X_i\right)_{i\in[n]}$ and let $\mathbf{\epsilon}=\left(\epsilon_1,\hdots,\epsilon_n\right)$ 
be a Rademacher sequence independent of $\mathcal{D}$. We introduce
\begin{align}
Z&:=\sum_{i\not= j}h(X_i,X_j)-\sum_{i\not= j}\epsilon_i\epsilon_jh(X_i,X_j)\label{eqnAux5}
\end{align}
and
$Y:=-\sum_{i\not= j}\epsilon_i\epsilon_jh(X_i,X_j)$.
Observe that
\begin{equation*}
\mathbb{E}\left[Z\left|\mathcal{D}\right.\right]=\sum_{i\not=j}h(X_i,X_j).
\end{equation*}
Let $\xi$ be an element of the dual space $B^*$. Conditional on $\mathcal{D}$, $\xi(Y)$ is a homogeneous scalar-valued Rademacher chaos of order $2$. 
By H\"older's inequality, we have for an arbitrary random variable $V$ with finite fourth moment that
\begin{align*}
\E\left[\left|V\right|^2\right]&\leq\left(\E\left[\left|V\right|\right]\right)^{1/2}\left(\E\left[\left|V\right|^3\right]\right)^{1/2}\\
&\leq\left(\E\left[\left|V\right|\right]\right)^{1/2}\left(\E\left[\left|V\right|^2\right]\right)^{1/4}\left(\E\left[\left|V\right|^4\right]\right)^{1/4}
\end{align*}
and therefore
\begin{equation}\label{eqnAux4}
\frac{\E\left[\left|V\right|^2\right]}{\left(\E\left[\left|V\right|^4\right]\right)^{1/2}}\leq\frac{\E\left[\left|V\right|\right]}{\left(\E\left[\left|V\right|^2\right]\right)^{1/2}}.
\end{equation}
Lemma $2.1$ from \cite{decoup_inequ} states that 
\begin{equation*}
\left(\mathbb{E}\left[\left.\left|\xi(Y)\right|^4\right|\mathcal{D}\right]\right)^{1/2}\leq 3\mathbb{E}\left[\left.\left|\xi(Y)\right|^2\right|\mathcal{D}\right].
\end{equation*}
Plugging this result into (\ref{eqnAux4}) gives
\begin{equation*}
\frac{\mathbb{E}\left[\left|\xi(Y)\right||\mathcal{D}\right]}{\left(\mathbb{E}\left[\left|\xi(Y)\right|^2|\mathcal{D}\right]\right)^{1/2}}\geq\frac{\mathbb{E}\left[\left|\xi(Y)\right|^2|\mathcal{D}\right]}{\left(\mathbb{E}\left[\left|\xi(Y)\right|^4|\mathcal{D}\right]\right)^{1/2}}\geq\frac{1}{3}.
\end{equation*}
Taking into account (\ref{eqnAux5}), an application of Lemma \ref{lemma_lin_form} yields 
\begin{equation}\label{eqnAux6}
\P\left(\left\|Z\right\|\geq\left\|\sum_{i\not= j}h(X_i,X_j)\right\|\bigg|\mathcal{D}\right)\geq\frac{1}{4}\left(\frac{1}{3}\right)^2=\frac{1}{36}.
\end{equation} 
Multiplying both sides of (\ref{eqnAux6}) by the characteristic function $\chi$ of the event \\$\left\{\left\|\sum_{i\not= j}h(X_i,X_j)\right\|>t\right\}$ 
and taking the expectation with respect to $\mathcal{D}$ gives 
\begin{align}
\P\left(\left\|\sum_{i\not= j}h(X_i,X_j)\right\|>t\right)&\leq36\P\left(\left\|Z\right\|>t\right)
=36 
\E_{\mathbf{\epsilon}}\left[\E\left[\chi_{\left\{\left\|Z\right\|>t\right\}}|\mathbf{\epsilon} \right]\right].\label{eqnAux7}
\end{align}
We conclude by noting that there is a vector $\epsilon^*\in\left\{\pm 1\right\}^n$ such that 
\begin{equation*}
\E\left[\chi_{\left\{\left\|Z\right\|>t\right\}}|\epsilon^*\right]\geq\E_{\mathbf{\epsilon}}\left[\E\left[\chi_{\left\{\left\|Z\right\|>t\right\}}|(\epsilon_1,\hdots,\epsilon_n)\right]\right].
\end{equation*}
The claim now follows by setting $S:=\left\{i\in\{1,\hdots,n\}|\epsilon_i^*=1\right\}$.\qed
\endproof
We will moreover need the following complex version of Hoeffding's inequality from \cite{nelteml11}, equation $(9)$.
\begin{theorem}\label{compl_Hoeffding} Let $\xi_1,\hdots,\xi_n$ be complex, independent and centered random variables satisfying $\left|\xi_k\right|\leq \alpha_k$ for constants $\alpha_1,\hdots,\alpha_n >0$. Set $\sigma^2:=\sum_{k=1}^n\alpha_k^2$. Then
\begin{equation}\label{eqn_compl_Hoeffding}
\P\left(\left|\sum_{k=1}^n\xi_k\right|>t\right)\leq 4\exp\left(-\frac{t^2}{4\sigma^2}\right).
\end{equation}
\end{theorem}
The final tool to prove that submatrices of $A$ are well-conditioned is the noncommutative Bernstein inequality from \cite{tro}.
\begin{theorem}\label{Bernstein_thm}
Let $X_1,\ldots,X_n\in\C^{s\times s}$ be a sequence of independent, mean zero and self-adjoint random matrices. Assume that, for some $K>0$, 
\begin{equation}\label{eqnfirst_Bern_cond}
\left\|X_{\ell}\right\|_{2\to 2}\leq K \quad \mbox{ a.s. for all }\; \ell \in [n] 
\end{equation}
and set 
\begin{equation}\label{eqnsecond_Bern_cond}
\sigma^2:=\left\|\sum_{\ell=1}^n\E X_{\ell}^2 \right\|_{2\to 2}.
\end{equation}
Then, for $t>0$, it holds that
\begin{equation}\label{eqnBernstein_inequ}
\P\left(\left\|\sum_{\ell=1}^n X_{\ell}\right\|_{2\to 2}\geq t\right)\leq 2s\exp\left(-\frac{t^2/2}{\sigma^2+Kt/3}\right).
\end{equation}
\end{theorem}

\begin{subsection}{Proof of Theorem \ref{central:thm}}

{\em }
Denote by 
\[
D_j:=\diag\left(\overline{\Phi_{\ell}(b_j)}: \ell\in T \right)\in\C^{s\times s}
\]
the diagonal matrix with diagonal consisting of the vector $\left(\overline{\Phi_{\ell}(b_j)}\right)_{\ell\in T}\in\C^s$ 
and introduce $g(b_k):=\left(\overline{\Phi_{\ell}(b_k)}\right)_{\ell\in T}\in\C^s$. 
Since $D_jD_j^*=\Id$ we observe that
\begin{equation}\label{eqn2}
\frac{1}{n^2}A_T^*A_T-\Id=\frac{1}{n^2}\sum_{j,k =1}^{n} \left[v(b_j,b_k)_T v(b_j,b_k)_T^* - \Id \right]=\frac{1}{n^2}\sum_{j=1}^nD_j\left(\sum_{k=1}^n\left[g(b_k)g(b_k)^*-\Id\right]\right)D_j^*. 
\end{equation}
Let $\mathbf{b}':=\left(b_1',\ldots,b_n'\right)$ denote an independent copy of $\mathbf{b}:=\left(b_1,\ldots,b_n\right)$. 
By the triangle inequality, we have
\begin{align}
\P\left(\left\|\frac{1}{n^2}A_{T}^*A_{T}-\Id\right\|_{2\to 2}\geq\delta\right)&\leq\P\left(\frac{1}{n^2}\left\|\sum_{j\not=k}\left[v(b_j,b_k)_{T}v(b_j,b_k)_{T}^*-\Id\right]\right\|_{2\to2}\geq\frac{\delta}{2}\right)\nonumber\\
&\qquad+\P\left(\frac{1}{n^2}\left\|\sum_{j=1}^n\left[v(b_j,b_j)_{T}v(b_j,b_j)_{T}^*-\Id\right]\right\|_{2\to2}\geq\frac{\delta}{2}\right)\nonumber
\end{align}
Using the decoupling inequality of Theorem \ref{Decoup_thm}, with $S\subset[n]$ denoting the corresponding set, and the symmetry relation $v(b_j,b_k)=v(b_k,b_j)$, we obtain for the first
term above
\begin{align}
&\P\left(\frac{1}{n^2}\left\|\sum_{j\not=k}\left[v(b_j,b_k)_{T}v(b_j,b_k)_{T}^*-\Id\right]\right\|_{2\to2}\geq\frac{\delta}{2}\right)\notag\\
\leq & 36\P\left(\frac{1}{n^2}\left\|\sum_{j\in S,k\in S^c}\left[v(b_j,b_k')_{T}v(b_j,b_k')_{T}^*-\Id\right]\right\|_{2\to2}\geq\frac{\delta}{8}\right).\label{eqnSndTerm}
\end{align}
We will now estimate the right hand side of (\ref{eqnSndTerm}).
Introducing
\begin{equation*}
X':=\sum_{k\in S^c}\left[g(b_k')g(b_k')^*-\Id\right]\in\C^{s\times s},
\end{equation*}
we observe that \eqref{eqn2} together with Fubini's theorem yields
\begin{align}
&36\P\left(\frac{1}{n^2}\left\|\sum_{j\in S,k\in S^c}\left[v(b_j,b_k')_{T}v(b_j,b_k')_{T}^*-\Id\right]\right\|_{2\to2}\geq\frac{\delta}{8}\right)\notag\\
= & 36\P\left(\frac{1}{n^2}\left\|\sum_{j\in S}D_jX'D_j^*\right\|_{2\to2}\geq\frac{\delta}{8}\right)=\E_{\mathbf{b}'}\left[36\P_{\mathbf{b}}\left(\frac{1}{n^2}\left\|\sum_{j\in S}D_jX'D_j^*\right\|_{2\to2}\geq\frac{\delta}{8}\right)\right]\label{eqn6}.
\end{align}
As the next step we apply the noncommutative Bernstein inequality, Theorem \ref{Bernstein_thm}, to the inner probability in (\ref{eqn6}).  
Since $D_j$ is a unitary matrix and the functions $\Phi_\ell$ are orthonormal we have
\begin{align}
\left\|D_jX'D_j^*\right\|_{2\to2}&=\left\|X'\right\|_{2\to 2},\label{eqn7}\\
\E\left[(D_j X'D_j^*)^2\right]&=\diag\left(X'^2\right),\nonumber
\end{align}
where $\diag\left(X'^2\right)$ denotes the matrix that coincides with $X'^2$ on the diagonal and is zero otherwise. Set $\mu$ to be the coherence parameter
\begin{equation*}
\mu:=
\max_{\ell,\tilde{\ell}\in T:\ell<\tilde{\ell}}\left|\sum_{k\in S^c}\Phi_\ell(b_k')\overline{\Phi_{\tilde{\ell}}(b_k')}\right|.
\end{equation*}
A crucial observation is that $\diag\left(X'^2\right)\preceq(s-1)\diag\left(\mu^2,\mu^2,\ldots,\mu^2\right)$, where $\preceq$ denotes the semidefinite ordering. Therefore, it holds that
\begin{equation}\label{eqn8}
\left\|\sum_{j\in S}\E\left[(D_j X'D_j^*)^2\right]\right\|_{2\to 2}\leq \left|S\right|(s-1)\mu^2\leq n(s-1)\mu^2.
\end{equation}
Plugging the bounds (\ref{eqn7}) and (\ref{eqn8}) into (\ref{eqnBernstein_inequ}) yields
\begin{equation}\label{eqn9}
\P_{\mathbf{b}}\left(\frac{1}{n^2}\left\|\sum_{j\in S}D_jX'D_j^*\right\|_{2\to2}\geq\frac{\delta}{8}\right)\leq2s\exp\left(-\frac{\delta^2}{\frac{128(s-1)}{n^3}\mu^2+\frac{16\delta}{3n^2}\left\|X'\right\|_{2\to2}}\right).
\end{equation}
Set $\tilde{\varepsilon} = \varepsilon/36$. Multiplying the inner probability in (\ref{eqn6}) by the characteristic function of the event $E:=E_1\cap E_2$, where
\begin{align*}
E_1&:=\left\{\frac{128(s-1)}{n^3}\mu^2\leq\frac{\delta^2}{2\log\left(8s/\tilde{\varepsilon}\right)}\right\},\\
E_2&:=\left\{\frac{16}{3n^2}\left\|X'\right\|_{2\to2}\leq\frac{\delta}{2\log\left(8s/\tilde{\varepsilon}\right)}\right\},
\end{align*}
we obtain, with $E_1^c$ and $E_2^c$ denoting the complements of $E_1$ and $E_2$,
\begin{equation}\label{eqn10}
36\,\P\left(\frac{1}{n^2}\left\|\sum_{j\in S}D_jX'D_j^*\right\|_{2\to2}\geq\frac{\delta}{8}\right)\leq \frac{\varepsilon}{4}+36\left(2s\P\left(E_1^c\right)+2s\P\left(E_2^c\right)\right).
\end{equation}
Therefore, it remains to estimate the probabilities of the events $E_1^c$ and $E_2^c$. For the event $E_1^c$, the union bound over all $s(s-1)/2\leq s^2/2$ two element subsets of $T$ implies in the case of a general BOS that 
\begin{align}
36\left(2s\P\left(E_1^c\right)\right)&\leq72s\P\left(\bigcup_{\ell,\tilde{\ell}\in T,\ell<\tilde{\ell}}\left\{\frac{128(s-1)}{n^3}\left|\sum_{k\in S^c}\Phi_\ell(b_k')\overline{\Phi_{\tilde{\ell}}(b_k')}\right|^2\geq\frac{\delta^2}{2\log\left(8s/\tilde{\varepsilon}\right)}\right\}\right)\nonumber\\&\leq72s\sum_{\ell,\tilde{\ell}\in T,\ell<\tilde{\ell}}\P\left(\left|\sum_{k\in S^c}\Phi_\ell(b_k')\overline{\Phi_{\tilde{\ell}}(b_k')}\right|\geq\frac{\delta n^{3/2}}{\sqrt{256s\log\left(8s/\tilde{\varepsilon}\right)}}\right)\label{eqn_pre_Hoeffding}\\&\leq144s^3\exp\left(-\frac{n^2\delta^2}{1024s\log\left(8s/\tilde{\varepsilon}\right)}\right),\label{eqn11}
\end{align}
where we have applied Hoeffding's inequality in the form of Theorem \ref{compl_Hoeffding} in the last line.
The right hand side of (\ref{eqn11}) is less than $\varepsilon/4$ provided
\begin{equation}\label{eqnAdd1}
n^2\geq1024\delta^{-2}s\log^2\left(576s^3/\varepsilon\right).
\end{equation}
As for $E_2^c$, we are going to apply the noncommutative Bernstein inequality again. Noting that 
\begin{align*}
\left\|g(b_k')g(b_k')^*-\Id\right\|_{2\to 2}&= s-1,\\
\left\|\sum_{k\in S^c}\E\left[\left(g(b_k')g(b_k')^*-\Id\right)^2\right]\right\|_{2\to 2}&=\left|S^c\right|(s-1)\leq n(s-1),
\end{align*}
we obtain 
\begin{equation}\label{eqn12}
36\left(2s\P\left(E_2^c\right)\right)\leq 144s^2\exp\left(-\frac{\delta^2}{\left(\frac{32}{3}\right)^2\frac{s}{n^3}\log^2\left(8s/\tilde{\varepsilon}\right)+\frac{32}{9}\delta\frac{s}{n^2}\log\left(8s/\tilde{\varepsilon}\right)}\right).
\end{equation}
Assuming (\ref{eqnAdd1}), the right hand side of (\ref{eqn12}) is less than $\varepsilon/4$. Since $\left\{\Phi_k^2\right\}_{k \in [N]}$ is also a BOS with respect to $(D,\nu)$, Condition (\ref{eqnAdd1}) implies, after another application of the noncommutative Bernstein inequality analogously to (\ref{eqn12}) and the preceding steps, that
\begin{equation}\label{eqnAdd2}
\P\left(\frac{1}{n^2}\left\|\sum_{j=1}^n\left[v(b_j,b_j)_{T}v(b_j,b_j)_{T}^*-\Id\right]\right\|_{2\to2}\geq\frac{\delta}{2}\right)\leq\frac{\varepsilon}{4}.
\end{equation}
This concludes the proof.\qed\endproof

\begin{remark}
\begin{itemize}
\item[(a)] In order to show (\ref{eqnAdd2}), we used the assumption that $\left\{\Phi_k^2\right\}_{k\in[N]}$ is also a BOS with respect to $(D,\nu)$. It might be that (\ref{eqnAdd2}) also holds under weaker assumptions on the BOS, however, it does not hold if we choose for example the Hadamard system. 
\item[(b)]\label{remark_ic} Assuming the special case of the scattering matrix (\ref{eqnRadar7}), the terms in (\ref{eqn_pre_Hoeffding}) take the form
\begin{equation*}
\Phi_\ell(b_k)\overline{\Phi_{\tilde{\ell}}(b_k)}=\exp\left(\frac{\pi i}{\lambda z_0}\left(\left\|r_\ell\right\|_2^2-\left\|r_{\tilde{\ell}}\right\|_2^2\right)\right)\exp\left(\frac{2\pi i}{\lambda z_0}\langle(r_{\tilde{\ell}}-r_\ell),b_k\rangle\right),
\end{equation*}
where due to the aperture condition (\ref{eqnRadar4})
\begin{equation*}
\tilde{\theta}_k:=\exp\left(\frac{2\pi i}{\lambda z_0}\langle(r_{\tilde{\ell}}-r_\ell),b_k\rangle\right)
\end{equation*}
is a Steinhaus random variable and $\tilde{\theta}:=(\tilde{\theta}_1,\hdots,\tilde{\theta}_n)$ is a Steinhaus sequence. We can therefore apply Hoeffding's inequality for Steinhaus sequences, see \cite{ra09}, Corollary 6.13. This inequality states that, for arbitrary $v\in\C^n$ and $\kappa\in(0,1)$,
\begin{equation}\label{eqnSteinhaus_Hoeffding}
\P\left(\left|\langle v,\tilde{\theta}\rangle\right|>t\right)\leq\frac{1}{1-\kappa}\exp\left(-\kappa \frac{t^2}{\left\|v\right\|_2^2}\right).
\end{equation}
Applying this result with $\kappa = 4/5$ instead of Theorem \ref{compl_Hoeffding} in (\ref{eqn_pre_Hoeffding}), one obtains that the claimed spectral norm estimate (\ref{eqnA18}) holds under the slightly improved condition 
\begin{equation}\label{eqn_improved_constants}
n^2\geq 320\delta^{-2}s\log\left(\frac{288s}{\varepsilon}\right)\log\left(\frac{720s^3}{\varepsilon}\right),
\end{equation} 
where we have also taken into consideration the precise form of (\ref{eqn11}).
\end{itemize}
\end{remark}
\end{subsection}

\end{section}

\begin{section}{Nonuniform Recovery of Scatterers with Random Phase}
\label{sec:nonuniform}

{\em Proof of Theorem \ref{Thm_rand_patterns}}.
The key idea of the proof is to apply Theorem~\ref{Thm_noisy_cond}.
Note first that 
\eqref{eq:BPDN} is equivalent to 
\begin{equation}\label{eqn_rand_rec_resc}
\argmin_{z\in\C^N}\left\|z\right\|_1 \quad \text{ subject to } \quad \left\|\frac{1}{n}Az-\frac{1}{n}y\right\|\leq \eta.
\end{equation}
Let $T \subset [N]$ be the index set corresponding to the $s$ largest absolute entries of $x$
and assume that $\sign(x)_T$ is either a Rademacher or a Steinhaus sequence. 
Suppose we are on the event
\begin{equation}\label{eqn_event}
E:= \left\{\left\|\frac{1}{n^2}A_{T}^*A_{T}-\Id\right\|_{2\to 2}\leq\frac{1}{2}\right\}.
\end{equation}
Theorem \ref{central:thm} states that $\P\left[E^c\right]\leq\varepsilon/2$ if 
\begin{equation}\label{eqnP1}
n^2 \geq4096s\log^2\left(1152s^3/\varepsilon\right).
\end{equation}
Set $\tilde{A}:=\frac{1}{n}A$. The event $E$ means in particular that
$\tilde{A}_T$ fulfills condition~\eqref{eqnN7}.
We define the vector $v\in\C^{n^2}$ in Theorem \ref{Thm_noisy_cond} via
\begin{equation}\label{eqnP2}
v:=\left(\tilde{A}^\dagger\right)^* \sign(x)_T = \tilde{A}_T\left(\tilde{A}_T^*\tilde{A}_T\right)^{-1}\sign(x)_T,
\end{equation}
where $\tilde{A}^\dagger$ denotes the pseudo-inverse of $\tilde{A}_T$.
Setting $u:=\tilde{A}^*v$, we have $u_T=\tilde{A}_T^*v=\sign(x)_T$, so that
(\ref{eqnN4}) is satisfied. Since we are on the event $E$, the smallest singular value of $\tilde{A}_T$ satisfies 
$\sigma_{\min}(\tilde{A}_T) \geq 1/\sqrt{2}$ and therefore
\[
\|v\|_2 \leq \|\tilde{A}^\dagger\|_{2 \to 2} \|\sgn(x)_T\|_2 \leq \sigma_{\min}(\tilde{A}_T)^{-1} \sqrt{s}
\leq \sqrt{2s}.
\]
Hence, also (\ref{eqnN6}) is satisfied. It remains to check (\ref{eqnN5}). To this end, note that
\begin{equation*}
\left\|u_{T^c}\right\|_\infty=\max_{\ell\in T^c}\left|\left\langle \left(\tilde{A}_T^*\tilde{A}_T\right)^{-1}\tilde{A}_T^*\tilde{a}_\ell,\sign(x)_T\right\rangle\right|=\max_{\ell\in T^c}\left|\langle \tilde{A}_T^{\dagger}\tilde{a}_\ell,\sign{x}_T\rangle\right|.
\end{equation*}
As in the previous section, we denote $\mathbf{b}=(b_1,\hdots,b_n)$. 
Since $\sign(x)_T=:\left(\theta_\ell\right)_{\ell\in T}=:\theta$ is a Rademacher or a Steinhaus sequence, condition (\ref{eqnP1}),
Fubini's Theorem and Hoeffding's inequality for Rademacher resp. Steinhaus sequences together with the union bound give 
\begin{align}
& \P\left(\max_{\ell\in T^c}\left|\langle \tilde{A}_T^{\dag}\tilde{a}_\ell,\sign(x)_T\rangle\right|>\frac{1}{2}\right) \leq\P\left(\left\{\max_{\ell\in T^c}\left|\langle \tilde{A}_T^{\dag}\tilde{a}_\ell,\sign(x)_T\rangle\right|>\frac{1}{2}\right\}\cap E\right)+\frac{\varepsilon}{2}\nonumber\\
\leq &\E_{\mathbf{b}}\left[\chi_{E}\sum_{\ell\in T^c}\P_{\theta}\left(\left|\langle \tilde{A}_T^{\dag}\tilde{a}_\ell,\sign(x)_T\rangle\right|>\frac{1}{2}\right)\right]+\frac{\varepsilon}{2}\nonumber
\\\leq&\E_{\mathbf{b}}\left[\chi_E\sum_{\ell\in T^c}2\exp\left(-\frac{1}{8\left\|\tilde{A}_T^{\dag}\tilde{a}_\ell\right\|_2^2}\right)\right]+\frac{\varepsilon}{2}\nonumber\\
\leq & 2(N-s)\E_{\mathbf{b}}\left[\chi_E\exp\left(-\frac{1}{8\max_{\ell\in T^c}\left\|\tilde{A}_T^{\dag}\tilde{a}_\ell\right\|_2^2}\right)\right]+\frac{\varepsilon}{2}.\label{eqnP4}
\end{align}
Since we are on the event $E$ from (\ref{eqn_event}), it follows as before that $\left\|\left(\tilde{A}_T^*\tilde{A}_T\right)^{-1}\right\|_{2\to 2}\leq\frac{1}{\sigma_{\text{min}}\left(\tilde{A}_T\right)^2}\leq 2$ and therefore 
\begin{equation*}
\max_{\ell\in T^c}\left\|\tilde{A}_T^{\dag}\tilde{a}_\ell\right\|_2^2\leq4 \max_{\ell \in T^c} \left\|\tilde{A}_T^*\tilde{a}_\ell\right\|_2^2\leq 4s\max_{\ell\in T^c,\tilde{\ell}\in T}\left|\langle\tilde{a}_{\ell},\tilde{a}_{\tilde{\ell}}\rangle\right|^2.
\end{equation*}
Set 
\begin{equation*}
\mu:=\max_{\ell\in T^c,\tilde{\ell}\in T}\left|\sum_{k=1}^n\Phi_\ell(b_k)\overline{\Phi_{\tilde{\ell}}(b_k)}\right|.
\end{equation*}
Since 
\begin{equation*}
\left|\langle\tilde{a}_{\ell},\tilde{a}_{\tilde{\ell}}\rangle\right|=\left|\sum_{k=1}^n\Phi_\ell(b_k)\overline{\Phi_{\tilde{\ell}}(b_k)}\right|^2,
\end{equation*}
we have
\begin{equation*}
\max_{\ell\in T^c}\left\|\tilde{A}_T^{\dag}\tilde{a}_\ell\right\|_2^2\leq4\frac{s}{n^4}\mu^4.
\end{equation*}
We then obtain 
\begin{align*}
&2(N-s)\E_{\mathbf{\mathbf{b}}}\left[\chi_E\exp\left(-\frac{1}{8\max_{\ell\in T^c}\left\|\tilde{A}_T^{\dag}\tilde{a}_\ell\right\|_2^2}\right)\right]\\&\leq 2(N-s)\E_{\mathbf{b}}\left[\chi_E\exp\left(-\frac{1}{32\frac{s}{n^4}\mu^4}\right)\right]\\
&\leq\frac{\varepsilon}{4}+2(N-s)\P_{\mathbf{b}}\left(\frac{s^{1/4}}{n}\mu>\frac{1}{\left(32\log\left(8(N-s)/\varepsilon\right)\right)^{1/4}}\right).
\end{align*}
Applying the union bound and Hoeffding's inequality as in (\ref{eqn11}) gives
\begin{align}
&2(N-s)\P_{\mathbf{b}}\left(\frac{s^{1/4}}{n}\mu>\frac{1}{\left(32\log\left(8(N-s)/\varepsilon\right)\right)^{1/4}}\right)\nonumber\\&\leq 8(N-s)^2s\exp\left(-\frac{n}{16\sqrt{2s\log\left(8(N-s)/\varepsilon\right)}}\right).\label{eqnfinal}
\end{align}
The condition 
\begin{equation}\label{eqnfinal1}
n\geq32\sqrt{s}\log^{1/2}\left(8(N-s)/\varepsilon\right)\log\left(576(N-s)^2s/\varepsilon\right)
\end{equation}
implies that the right hand side of equation (\ref{eqnfinal}) is less than $\varepsilon/4$. Assuming $s\leq N/3$ and $8(N-s)/\varepsilon\geq e^4$, (\ref{eqnfinal1}) implies (\ref{eqnP1}) and therefore also $\P\left(E^c\right)\leq\varepsilon/2$, where $E$ is the event from (\ref{eqn_event}).
We have thus verified that under condition (\ref{eqnfinal1}), all conditions of Theorem \ref{Thm_noisy_cond} are satisfied with probability at least $1-\varepsilon$. 
Since we work with the rescaled version (\ref{eqn_rand_rec_resc}) of $A$, the solution $\hat{x}$ satisfies (\ref{eqn_rand_rec_err}) with the required probability. 
This finishes the proof of Theorem \ref{Thm_rand_patterns}.\qed
\endproof
\begin{remark} In the special case of the scattering matrix (\ref{eqnRadar7}), we can apply the same technique as in Remark \ref{remark_ic} (b) to obtain a slight improvement of (\ref{eqnfinal1}). In fact, assuming also the mild condition $8(N-s)/\varepsilon\geq e^7$, all conditions of Theorem \ref{Thm_noisy_cond} are satisfied with probability at least $1-\varepsilon$ under the improved condition
\begin{equation*}
n\geq 5\sqrt{2s}\log^{1/2}\left(8(N-s)/\varepsilon\right)\log\left(576s(N-s)^2/\varepsilon\right).
\end{equation*}
\end{remark}
\end{section}

\begin{section}{Nonuniform Recovery of Scatterers with Deterministic Phase}
\label{sec:deterministic}

\begin{subsection}{Set partitions}\label{section_set_partitions}

To prove the central result of this section, we will require a few facts on certain partitions of the set $[N]$, $N\in \N$. 
As in \cite[Section 2.2]{ra05-7} we define $\mathcal{P}\left(N,k\right)$ as the set of all partitions of $[N]$ into exactly $k$ blocks such that 
each block contains at least two elements. Note that then necessarily $k\leq N/2$. The numbers $S_2(N,k):=\left|\mathcal{P}(N,k)\right|$ are called 
associated Stirling numbers of the second kind. In \cite[Section 3.5]{ra05-7} it was shown that
\begin{equation}\label{eqnA21}
S_2(N,k)\leq\left(\frac{3N}{2}\right)^{N-k}.
\end{equation} 
For our purposes, we will also need partitions of $[N]$ in which not necessarily all blocks contain at least two elements. 
\begin{definition}
For $N\geq 1$, $t\leq k\leq N$, we define $\mathcal{P}\left(N,k,k-t\right)$ as the set of all partitions of $[N]$ into $k$ blocks such that $k-t$ of these blocks contain at least two elements. Moreover, we define $\mathcal{P}_{ex}\left(N,k,k-t\right)$ as the set of all partitions of $[N]$ into $k$ blocks such that exactly $k-t$ blocks contain at least two elements and exactly $t$ blocks contain exactly one element.
\end{definition}

The above definition of $\mathcal{P}\left(N,k,k-t\right)$ implies that necessarily
$2(k-t)\leq N-t$
and therefore 
\begin{equation}\label{eqnA22}
k\leq \frac{N+t}{2}.
\end{equation}
Our next goal is a convenient estimate of the numbers $S\left(N,k,k-t\right):=\left|\mathcal{P}\left(N,k,k-t\right)\right|$. 
We first observe that 
\begin{equation*}
S\left(N,k,k-t\right)=\sum_{r=0}^t\left|\mathcal{P}_{ex}\left(N,k,k-r\right)\right|.
\end{equation*}
Moreover, we have
\begin{equation}
\left|\mathcal{P}_{ex}\left(N,k,k-r\right)\right|=\binom{N}{r} S_2(N-r,k-r)\leq \binom{N}{r} \left(\frac{3N}{2}\right)^{N-k},
\end{equation} 
where the last inequality follows from the estimate (\ref{eqnA21}). Since $t\leq N$  and therefore $\sum_{r=0}^t\binom{N}{r}\leq 2^N$, this yields
\begin{equation}\label{eqnA24}
S\left(N,k,k-t\right)\leq (3N)^N\left(\frac{3N}{2}\right)^{-k}.
\end{equation}
This estimate will become crucial in the next section.

\end{subsection}
\begin{subsection}{Construction of a dual certificate}

We will use combinatorial estimates inspired by the analysis in \cite{carota06,pfra10,ra05-7,cata10} in order to construct
a dual certificate. Hereby, we exploit the estimates on set partitions stated above.
In this way, we will extend the recovery result of Section \ref{sec:random_sign} to a vector $x\in\C^N$ 
with deterministic phase pattern $\sign(x)_T$ -- recall that $T$ is the set of indices corresponding to the $s$ largest absolute entries
of $x$. 
Since the phases are now deterministic we can no longer use the additional concentration of measure 
coming from the independent randomness in the signs. In particular, we have to estimate the probability of the event
\begin{equation*}
\left\{\max_{\ell\in T^c}\left|\langle \tilde{A}_T^{\dag}\tilde{a}_\ell,\sign(x)_T\rangle\right|>\frac{1}{2}\right\}
\end{equation*} 
using only the randomness in $\tilde{A}$.
Throughout this subsection, we will assume that the sampling matrix $A\in\C^{n^2\times N}$ is given by (\ref{eqnRadar7}). 
However, we note that exactly the same proof applies if we take the Fourier system $\{\Phi_k\}$ 
from \cite{ra05-7} instead and construct the random matrix as in \eqref{eqn_BOS_sample}. 

Let us state the central result of this section.
\begin{theorem}\label{thm_nonuniform_deterministic} Let $A\in\C^{n^2\times N}$ be the random sampling matrix from (\ref{eqnRadar7}) and let $x\in\C^N$. 
Let $T\subset[N]$ with $\left|T\right|=s$ be the index set of the $s$ largest absolute entries of $x$. 
Set $\tilde{A}:=\frac{1}{n}A$. 
If 
\begin{equation}\label{eqnA25}
n^2\geq Cs\log^2\left(cN/\varepsilon\right),
\end{equation}
then 
with probability at least $1-\varepsilon$ 
\begin{enumerate}
\item[(i)] there is a $v\in\C^{n^2}$ such that $u:=\tilde{A}^*v$ and $v$ satisfy Conditions (\ref{eqnN4}),(\ref{eqnN5}) and (\ref{eqnN6}) 
of Theorem \ref{Thm_noisy_cond};
\item[(ii)] for the matrix $\tilde{A}$, it holds that 
\begin{equation}\label{eqnD1}
\left\|\tilde{A}_{T}^*\tilde{A}_{T}-\Id\right\|_{2\to 2}\leq\frac{1}{e}.
\end{equation}
\end{enumerate}
The constants satisfy $C\leq\left(800 e^{3/4}\right)^2$, $c\leq 6$.
\end{theorem}
\begin{proof} Suppose we are on the event 
\begin{equation*}
E:=
\left\{\left\|\tilde{A}_{T}^*\tilde{A}_{T}-\Id\right\|_{2\to 2}\leq\frac{1}{e}\right\},
\end{equation*} 
where the constant $1/e$ in the probability is chosen to ease computations later on.  
Theorem \ref{central:thm} implies that $\P\left[E^c\right]\leq\varepsilon/4$ if Condition \eqref{eqnA25} holds.
Our aim is an estimate for the probability of the event 
\begin{equation}\label{eqnA26}
\widetilde{E}:=\left\{\left\|\tilde{A}_{T^c}^*\tilde{A}_T\left(\tilde{A}_T^*\tilde{A}_T\right)^{-1}\sign(x)_T\right\|_\infty> \frac{1}{2}\right\}.
\end{equation} 
By expanding the Neumann series, we observe that, for $m\in \N$,
\begin{equation*}
\left(\Id-\left(\Id-\tilde{A}_T^*\tilde{A}_T\right)^m\right)^{-1}=\Id+\sum_{r=1}^\infty\left(\Id-\tilde{A}_T^*\tilde{A}_T\right)^{rm}.
\end{equation*}
With
\begin{equation*}
A_m:=\sum_{r=1}^\infty\left(\Id-\tilde{A}_T^*\tilde{A}_T\right)^{rm}
\end{equation*}
we obtain
\begin{align*}
\left(\tilde{A}_T^*\tilde{A}_T\right)^{-1}&=\left(\Id-\left(\Id-\tilde{A}_T^*\tilde{A}_T\right)\right)^{-1}
=\left(\Id-\left(\Id-\tilde{A}_T^*\tilde{A}_T\right)^m\right)^{-1}\sum_{k=0}^{m-1}\left(\Id-\tilde{A}_T^*\tilde{A}_T\right)^k\\&=\left(\Id+A_m\right)\sum_{k=0}^{m-1}\left(\Id-\tilde{A}_T^*\tilde{A}_T\right)^k.
\end{align*}
An application to $\sign(x)_T$ yields
\begin{align*}
\tilde{A}_{T^c}^*\tilde{A}_T\left(\tilde{A}_T^*\tilde{A}_T\right)^{-1}\sign(x)_T&=\tilde{A}_{T^c}^*\tilde{A}_T\sum_{k=0}^{m-1}\left(\Id-\tilde{A}_T^*\tilde{A}_T\right)^k\sign(x)_T\\&+\tilde{A}_{T^c}^*\tilde{A}_T A_m\sum_{k=0}^{m-1}\left(\Id-\tilde{A}_T^*\tilde{A}_T\right)^k\sign(x)_T.
\end{align*}
An application of the pigeon hole principle yields
\begin{align}
\P\left(\widetilde{E}\right)&\leq\P\left(\left\|\tilde{A}_{T^c}^*\tilde{A}_T\sum_{k=0}^{m-1}\left(\Id-\tilde{A}_T^*\tilde{A}_T\right)^k\sign(x)_T\right\|_\infty>\frac{1}{4}\right)\label{eqnA27}\\&+\P\left(\left\|\tilde{A}_{T^c}^*\tilde{A}_T A_m\sum_{k=0}^{m-1}\left(\Id-\tilde{A}_T^*\tilde{A}_T\right)^k\sign(x)_T\right\|_\infty>\frac{1}{4}\right).\label{eqnA28}
\end{align}
We now choose \begin{equation}\label{eqn_m_cond}
m:=\left\lceil 2\log\left(6N/\varepsilon\right)\right\rceil.
\end{equation} 
For the treatment of the event 
\begin{equation}\label{def:barE}
\ol{E}:=\left\{\left\|\tilde{A}_{T^c}^*\tilde{A}_T A_m\sum_{k=0}^{m-1}\left(\Id-\tilde{A}_T^*\tilde{A}_T\right)^k\sign(x)_T\right\|_\infty>\frac{1}{4}\right\},
\end{equation}
in (\ref{eqnA28}) we denote by $a_\ell$ the columns 
of the unnormalized sampling matrix $A$ and set 
\begin{equation*}
\mu^2:=\max_{\ell\in T^c,\tilde{\ell}\in T}\left|\langle a_\ell,a_{\tilde{\ell}}\rangle\right|.
\end{equation*}
For a matrix $B\in\C^{m\times k}$, we denote by
\begin{equation*}
\left\|B\right\|_{\infty\to\infty}:=\sup_{\left\|x\right\|_\infty=1}\left\|Bx\right\|_\infty=\max_{\ell\in[m]}\sum_{n\in[k]}|b_{\ell n}|
\end{equation*}
the operator norm of $B$ on $\ell_\infty$.
We then obtain
\begin{equation*}
\left\|\tilde{A}_{T^c}^*\tilde{A}_T\right\|_{\infty\to\infty}\leq\frac{s}{n^2}\mu^2.
\end{equation*}
Moreover, for an arbitrary matrix $B\in\C^{s\times s}$, it follows from the definition of $\left\|\cdot\right\|_{\infty\to\infty}$ that 
$\left\|B\right\|_{\infty\to\infty}\leq\sqrt{s}\left\|B\right\|_{2\to2}$.
Conditionally on the event $E$, this inequality gives 
\begin{align*}
\left\|A_m\right\|_{\infty\to\infty}&\leq\sqrt{s}\left\|A_m\right\|_{2\to 2}
\leq\sqrt{s}\sum_{r=1}^\infty\left\|\left(\Id-\tilde{A}_T^*\tilde{A}_T\right)\right\|_{2\to 2}^{rm} 
\leq\sqrt{s}\sum_{r=1}^\infty\left(\frac{1}{e^{m}}\right)^r
=\sqrt{s}\frac{1}{e^m-1}.
\end{align*}
Similarly, we obtain
\begin{equation*}
\left\|\sum_{k=0}^{m-1}\left(\Id-\tilde{A}_T^*\tilde{A}_T\right)^k\right\|_{\infty\to\infty}\leq\sqrt{s}\frac{e}{e-1}.
\end{equation*}
Combining these estimates, we obtain, conditionally on the event $E$,
\begin{align*}
&\left\|\tilde{A}_{T^c}^*\tilde{A}_T A_m\sum_{k=0}^{m-1}\left(\Id-\tilde{A}_T^*\tilde{A}_T\right)^k\sign(x)_T\right\|_{\infty\to\infty}\\
\leq & \left\|\tilde{A}_{T^c}^*\tilde{A}_T\right\|_{\infty\to\infty}\left\|A_m\right\|_{\infty\to\infty}
\left\|\sum_{k=0}^{m-1}\left(\Id-\tilde{A}_T^*\tilde{A}_T\right)^k\right\|_{\infty\to\infty}\\
\leq &\frac{s^2}{n^2} \frac{e}{(e-1)}\frac{1}{e^m-1}\mu^2\leq 4\frac{s^2}{e^m}\frac{1}{n^2}\mu^2\leq\frac{\varepsilon^2}{9n^2}\mu^2,
\end{align*}
where we have applied (\ref{eqn_m_cond}) and the fact that $s\leq N$ in the last line.
%
Hence, 
the probability of the event $\ol{E}$ in \eqref{def:barE} can be bounded by
\begin{align*}
\P\left(\ol{E}\right)&=\P\left(\ol{E}\cap E\right)+\P\left(\ol{E}\cap E^c\right) 
\leq\P\left(\frac{\varepsilon^2}{9n^2}\mu^2>\frac{1}{4}\right)+\frac{\varepsilon}{4}\\&\leq 4s(N-s)\exp\left(-\frac{9n}{8\varepsilon^2}\right)+\frac{\varepsilon}{4} 
\leq\frac{\varepsilon}{2},
\end{align*}
where we have applied Hoeffding's inequality Theorem \ref{compl_Hoeffding} and the union bound together with (\ref{eqnA25}) in the last line.
It remains to estimate the term in (\ref{eqnA27}). To this end, we define, for  $\ell\in T^c$, 
\begin{equation}\label{eqnA31}
E_\ell:=\left\{\left|\sum_{k=0}^{m-1}\tilde{a}_\ell^*\tilde{A}_T\left(\Id-\tilde{A}_T^*\tilde{A}_T\right)^k\sign(x)_T\right|>\frac{1}{4}\right\}.
\end{equation}
Let $\left\{\beta_k\right\}_{k=0,\ldots,m-1}\subset(0,1)$ such that $\sum_{k=0}^{m-1}\beta_k\leq1/4$ and let $M_k \in \N$ to be chosen below.
According to the pigeon hole principle, 
we have
\begin{align*}
\P\left(E_\ell\right)&\leq\sum_{k=0}^{m-1}\P\left(\left|\tilde{a}_\ell^*\tilde{A}_T\left(\Id-\tilde{A}_T^*\tilde{A}_T\right)^k\sign(x)_T\right|\geq\beta_k\right)\\
&=\sum_{k=0}^{m-1}\P\left(\left|\tilde{a}_\ell^*\tilde{A}_T\left(\Id-\tilde{A}_T^*\tilde{A}_T\right)^k\sign(x)_T\right|^{2M_k}\geq\beta_k^{2M_k}\right)\\
&\leq\sum_{k=0}^{m-1}\beta_k^{-2M_k}\E\left[\left|\tilde{a}_\ell^*\tilde{A}_T\left(\Id-\tilde{A}_T^*\tilde{A}_T\right)^k\sign(x)_T\right|^{2M_k}\right],
\end{align*}
where we have applied Markov's inequality in the last step.
With $r(\cdot)$ denoting the function that rounds to the closest integer,
we introduce
\begin{align*}
M_k:=r\left(\frac{m}{k+1}\right)\text{ for }k=0,\ldots,m-1, \quad
q_k:=2M_k(k+1).
\end{align*}
Then $2m/3\leq M_k(k+1)\leq 4m/3$ and therefore $4m/3\leq q_k\leq 8m/3$ 
and also $m/M_k\geq 3(k+1)/4$. For some $\beta\in(0,1)$, we further set
\begin{equation*}
\beta_k:=\beta^{\frac{m}{M_k}}.
\end{equation*}
Then with $\beta=1/(5^{4/3})$, we have $\sum_{k=0}^{m-1}\beta_k\leq1/4$, so that we have found valid choices for the $\beta_k$. 
The rest of the proof 
is a straightforward consequence of the following statement.
\begin{lemma}\label{Lemma4} Let $k,M\in\N$ be given and set $q=2M(k+1)$.
If
\begin{equation}\label{eqnA38}
n\geq3q\sqrt{s},
\end{equation} 
then
\begin{equation}\label{eqnA29}
\E\left[\left|\tilde{a}_\ell^*\tilde{A}_T\left(\Id-\tilde{A}_T^*\tilde{A}_T\right)^k\sign(x)_T\right|^{2M}\right]\leq6q\left(\frac{6q\sqrt{s}}{n}\right)^q.
\end{equation}
\end{lemma}

Before we prove Lemma \ref{Lemma4}, let us first see how one can deduce Theorem \ref{thm_nonuniform_deterministic} from it.
Condition \eqref{eqnA25} implies
\begin{equation*}
n\geq 800e^{3/4}\sqrt{s}\log\left(\frac{6N}{\varepsilon}\right),
\end{equation*} 
which, according to the choice $m =\left\lceil 2\log\left(6N/\varepsilon\right)\right\rceil $ of $m$ and the definition of $q$ implies (\ref{eqnA38}).
Then \eqref{eqnA29} yields the series of inequalities
\begin{align*}
&\sum_{k=0}^{m-1}\beta_k^{-2M_k}\E\left[\left|\tilde{a}_\ell^*\tilde{A}_T\left(\Id-\tilde{A}_T^*\tilde{A}_T\right)^k\sign(x)_T\right|^{2M_k}\right] \leq\beta^{-2m}\sum_{k=0}^{m-1}6q_k\left(\frac{6q_k\sqrt{s}}{n}\right)^{q_k}\\
\leq & \beta^{-2m}\sum_{k=0}^{m-1}16m\left(\frac{16m\sqrt{s}}{n}\right)^{\frac{4}{3}m} \leq16m^2\left(\frac{16\beta^{-3/2}m\sqrt{s}}{n}\right)^{\frac{4}{3}m}.
\end{align*}
With $E_\ell$ denoting the events from (\ref{eqnA31}), we further obtain, using (\ref{eqnA25}) once more,
\begin{align*}
\sum_{\ell\notin T}\P\left[E_\ell\right]&\leq16(N-s)m^2\left(\frac{16\beta^{-3/2}m\sqrt{s}}{n}\right)^{\frac{4}{3}m} \leq16(N-s)m^2e^{-m}\leq\frac{\varepsilon}{2}.
\end{align*}
This finishes the proof of Theorem \ref{thm_nonuniform_deterministic}.
\end{proof}
What remains is the following\\
{\em Proof of Lemma \ref{Lemma4}}. So far, we have not used that the bounded orthonormal system underlying the random scattering matrix 
has the specific structure defined in (\ref{eqnRadar7}). 
In what follows, we will use the letter $\ell\in\Z^2$, possibly indexed further, 
to denote the rescaled positions (without the distance coordinate) on the resolution grid where the targets can be. We furthermore identify $[N]$ with $[\sqrt{N}]^2$ in the canonical way, thereby recovering the square grid of resolution cells (recall that we set $N:=\lfloor 2L/d_0\rfloor^2$, where $L>0$ is the size of the target domain and $d_0>0$ denotes the meshsize of the resolution grid, so that $\sqrt{N}$ is actually the number of resolution cells along one axis of the square array). We fix $\ell\in T^c$ and set 
 $\ell_0^{(h)}:=\ell$ for $h=1,\ldots,2M$. A lengthy but straightforward calculation gives with $\omega:=2\pi d_0/(\lambda z_0)$
\begin{align}
&\left|\tilde{a}_\ell^*\tilde{A}_T\left(\Id-\tilde{A}_T^*\tilde{A}_T\right)^k\sign(x)_T\right|^{2M}\notag=\frac{1}{n^{4M(k+1)}}\\
&\times\sum_{\substack{j_1^{(1)},\ldots,j_{k+1}^{(1)}=1\\\vdots\\j_1^{(2M)},\ldots,j_{k+1}^{(2M)}=1}}^n\sum_{\substack{m_1^{(1)},\ldots,m_{k+1}^{(1)}=1\\\vdots\\m_1^{(2M)},\ldots,m_{k+1}^{(2M)}=1}}^n\sum_{\substack{\ell_1^{(1)},\ldots,\ell_{k+1}^{(1)}\in T\\\vdots\\\ell_1^{(2M)},\ldots,\ell_{k+1}^{(2M)}\in T\\\ell_h^{(p)}\not=\ell_{h-1}^{(p)},h\in[k+1]}} 
\prod_{t=1}^{M}  \sign\left(x_{\ell_{k+1}^{(2t-1)}}\right)\overline{\sign\left(x_{\ell_{k+1}^{(2t)}}\right)} \nonumber\\
& \quad  \times \exp\left(i\frac{\omega}{2}\sum_{p=1}^{2M}(-1)^p\left\|\ell_{k+1}^{(p)}\right\|_2^2\right) 
\;\exp\left(i\omega\sum_{p=1}^{2M}(-1)^p\sum_{h=1}^{k+1}\left\langle\left(\ell_{h-1}^{(p)}-\ell_h^{(p)}\right),b_{j_h^{(p)}}\right\rangle\right)\nonumber\\
& \quad \times\exp\left(i\omega\sum_{p=1}^{2M}(-1)^p\sum_{h=1}^{k+1}\left\langle\left(\ell_{h-1}^{(p)}-\ell_h^{(p)}\right),b_{m_h^{(p)}}\right\rangle\right).\label{eqnA33}
\end{align}
In order to evaluate the above term, we will use combinatorial arguments
inspired by \cite{carota06,ra05-7}. To a given word
$\left(j_h^{(p)}\right)_{h=1,\ldots,k+1}^{p=1,\ldots,2M}$ we associate the
partition $\mathcal{Q}$ of $[k+1]\times[2M]$ with the property that $(h,p)$
and $(h',p')$ are in the same block if and only if
$j_{h}^{(p)}=j_{h'}^{(p')}$. Analogously, we associate the partition
$\mathcal{R}$ to the word
$\left(m_h^{(p)}\right)_{h=1,\ldots,k+1}^{p=1,\ldots,2M}$. 
To each $Q\in\mathcal{Q}$ resp.\ $R\in\mathcal{R}$ there exists exactly one
$j_Q\in\left\{1,\ldots,n\right\}$ resp.\ $m_R\in\left\{1,\ldots,n\right\}$
such that $j_{h}^{(p)}=j_Q$ for all $(h,p)\in Q$ resp.\ $m_{h}^{(p)}=m_R$ for all $(h,p)\in R$. 
We define
\begin{align*}
\Qpar\cap\Rpar&:=\left\{(Q,R)\in\Qpar\times\Rpar : j_Q=m_R\right\},\\
\Qpar^\cap&:=\left\{Q\in\Qpar : \text{there exists } R=R(Q)\in\Rpar \text{ such that }m_{R(Q)}=j_Q\right\},\\
\Rpar^\cap&:=\left\{R\in\Rpar : \text{there exists }Q=Q(R)\in\Qpar \text{ such that }j_{Q(R)}=m_R\right\}.
\end{align*}
With this notation, we can write
\begin{align*}
&\E\left[\exp\left(i\omega\sum_{p=1}^{2M}(-1)^p\sum_{h=1}^{k+1}\left\langle\left(\ell_{h-1}^{(p)}-\ell_h^{(p)}\right),b_{j_h^{(p)}}\right\rangle\right)\right.\\
& \phantom{\E((}\left. \times \exp\left(i\omega\sum_{p=1}^{2M}(-1)^p\sum_{h=1}^{k+1}\left\langle\left(\ell_{h-1}^{(p)}-\ell_h^{(p)}\right),b_{m_h^{(p)}}\right\rangle\right)\right]\\
= &\E\left[\prod_{Q\in\Qpar\setminus\Qpar^\cap}\exp\left(i\omega\left\langle\sum_{(h,p)\in Q}(-1)^p\left(\ell_{h-1}^{(p)}-\ell_h^{(p)}\right),b_{j_Q}\right\rangle\right)\right]\\
\times&\E\left[\prod_{R\in\Rpar\setminus\Rpar^\cap}\exp\left(i\omega\left\langle\sum_{(h,p)\in R}(-1)^p\left(\ell_{h-1}^{(p)}-\ell_h^{(p)}\right),b_{m_R}\right\rangle\right)\right]\\
\times &\E\left[\prod_{Q\in\Qpar^\cap}\exp\left(i\omega\left\langle\sum_{(h,p)\in Q}(-1)^p\left(\ell_{h-1}^{(p)}-\ell_h^{(p)}\right),b_{j_Q}\right\rangle \right.\right.\\
& \phantom{\times\E[[\prod_{Q\in\Qpar^\cap}\exp((}\left.\left. +i\omega\left\langle\sum_{(h,p)\in R(Q)}(-1)^p\left(\ell_{h-1}^{(p)}-\ell_h^{(p)}\right),b_{m_{R(Q)}}\right\rangle\right)\right].
\end{align*}
Observe that
\begin{align*}
&\E\left[\prod_{Q\in\Qpar\setminus\Qpar^\cap}\exp\left(i\omega\left\langle\sum_{(h,p)\in Q}(-1)^p\left(\ell_{h-1}^{(p)}-\ell_h^{(p)}\right),b_{j_Q}\right\rangle\right)\right]\\&=\prod_{Q\in\Qpar\setminus\Qpar^\cap}\delta\left(\sum_{(h,p)\in Q}(-1)^p\left(\ell_{h-1}^{(p)}-\ell_h^{(p)}\right)\right),
\end{align*}  
where $\delta$ is the Kronecker delta, that is $\delta(0)=1$ and $\delta(x)=0$ for $x\not= 0$.
Since $\ell_h^{(p)}\not=\ell_{h-1}^{p}$, this implies that each $Q\in\Qpar\setminus\Qpar^\cap$ must contain at least two elements in order to provide a nonzero contribution 
to the overall expectation of the expression in \eqref{eqnA33}. 
The same is true for each $R\in\Rpar\setminus\Rpar^\cap$. However, the blocks $Q\in\Qpar^\cap$ may contain just one element, since they have a 
corresponding block $R(Q)$ with matching index. Therefore, we can break
the evaluation of the right hand side of (\ref{eqnA33}) down to three basic questions.
\begin{enumerate}
\item[1.] What are the numbers $t_1$ resp.\ $t_2$ of the distinct indices
appearing in the words
$w_1:=\left(j_h^{(p)}\right)_{h=1,\ldots,k+1}^{p=1,\ldots,2M}$ resp.\ $w_2:=\left(m_h^{(p)}\right)_{h=1,\ldots,k+1}^{p=1,\ldots,2M}$?
\item[2.] What is the number $t$ of indices that the words $w_1$ and $w_2$ have in common?
\item[3.] Given $1.$ and $2.$, which constraints must be fulfilled by the partitions $\Qpar$ and $\Rpar$ corresponding to $w_1$ and $w_2$? 
\end{enumerate}
In the following, we identify partitions of $[k+1]\times[2M]$ with partitions of $[2M(k+1)]$ in the canonical way. 
Moreover, if we have a partition $\Qpar=\left\{Q_1,\ldots,Q_t,Q_{t+1},\ldots,Q_{t_1}\right\}$, we enumerate it without loss of generality 
such that $Q_{t+1},\ldots,Q_{t_1}$ are the blocks containing at least two elements and $Q_1,\ldots,Q_t$ are the blocks which might contain just one element. The same is done for the partition $\Rpar=\left\{R_1,\ldots,R_t,R_{t+1},\ldots,R_{t_2}\right\}$. We define
\begin{equation*}
{\cal E} :=\E\left[\left|\tilde{a}_\ell^*\tilde{A}_T\left(\Id-\tilde{A}_T^*\tilde{A}_T\right)^k\sign(x)_T\right|^{2M}\right].
\end{equation*}
Using the triangle inequality and $n>2M(k+1)$ implied by (\ref{eqnA38}) together with the definitions from Subsection \ref{section_set_partitions} we obtain
\begin{align}
{\cal E} &\leq\frac{1}{n^{4M(k+1)}}\sum_{t=0}^{2M(k+1)}\sum_{t_1=t}^{M(k+1)+\left\lfloor t/2\right\rfloor}\sum_{t_2=t}^{M(k+1)+\left\lfloor t/2\right\rfloor}\sum_{\substack{j_1,\ldots,j_{t_1}\text{ pw different }\\m_1,\ldots,m_{t_2}\text{ pw different }\\\left|\left\{j_1,\ldots,j_{t_1}\right\}\cap\left\{m_1,\ldots,m_{t_2}\right\}\right|=t}}\nonumber\\&\sum_{\Qpar\in\mathcal{P}\left(2M(k+1),t_1,t_1-t\right)}\nonumber\sum_{\Rpar\in\mathcal{P}\left(2M(k+1),t_2,t_2-t\right)}\sum_{\substack{\ell_1^{(1)},\ldots,\ell_{k+1}^{(1)}\in T\\\vdots\\\ell_1^{(2M)},\ldots,\ell_{k+1}^{(2M)}\in T\\\ell_h^{(p)}\not=\ell_{h-1}^{p},h\in[k+1]}}\nonumber\\&\prod_{Q\in\left\{Q_{t+1},\ldots,Q_{t_1}\right\}}\delta\left(\sum_{(h,p)\in Q}(-1)^p\left(\ell_{h-1}^{(p)}-\ell_h^{(p)}\right)\right)\label{eqnA34}\\
&\times\prod_{R\in\left\{R_{t+1},\ldots,R_{t_2}\right\}}\delta\left(\sum_{(h,p)\in R}(-1)^p\left(\ell_{h-1}^{(p)}-\ell_h^{(p)}\right)\right)\label{eqnA35}\\
&\times\prod_{j=1}^t\delta\left(\sum_{(h,p)\in Q_j}(-1)^p\left(\ell_{h-1}^{(p)}-\ell_h^{(p)}\right)+\sum_{(h,p)\in R_j}(-1)^p\left(\ell_{h-1}^{(p)}-\ell_h^{(p)}\right)\right).\label{eqnA36}
\end{align}
For the product $\prod_{Q\in\left\{Q_{t+1},\ldots,Q_{t_1}\right\}}\delta\left(\sum_{(h,p)\in Q}(-1)^p\left(\ell_{h-1}^{(p)}-\ell_h^{(p)}\right)\right)$
to be nonzero, we must have 
$
\sum_{(h,p)\in Q}(-1)^p\left(\ell_{h-1}^{(p)}-\ell_h^{(p)}\right)=0\text{ for all }Q\in\left\{Q_{t+1},\ldots,Q_{t_1}\right\},
$ 
and analogously for the other two products appearing in (\ref{eqnA35}), (\ref{eqnA36}).
Therefore, the expressions (\ref{eqnA34})-(\ref{eqnA36}) give at least $t_1\vee t_2:=\max\{t_1,t_2\}$ linearly independent constraints. Recalling that $q=2M(k+1)$, this observation yields
\begin{align*}
&\sum_{\substack{\ell_1^{(1)},\ldots,\ell_{k+1}^{(1)}\in T\\\vdots\\\ell_1^{(2M)},\ldots,\ell_{k+1}^{(2M)}\in T\\\ell_h^{(p)}\not=\ell_{h-1}^{p},h\in[k+1]}}\prod_{Q\in\left\{Q_{t+1},\ldots,Q_{t_1}\right\}}\delta\left(\sum_{(h,p)\in Q}(-1)^p\left(\ell_{h-1}^{(p)}-\ell_h^{(p)}\right)\right)
\\&\times\prod_{R\in\left\{R_{t+1},\ldots,R_{t_2}\right\}}\delta\left(\sum_{(h,p)\in R}(-1)^p\left(\ell_{h-1}^{(p)}-\ell_h^{(p)}\right)\right)\\
&\times\prod_{j=1}^t\delta\left(\sum_{(h,p)\in Q_j}(-1)^p\left(\ell_{h-1}^{(p)}-\ell_h^{(p)}\right)+\sum_{(h,p)\in R_j}(-1)^p\left(\ell_{h-1}^{(p)}-\ell_h^{(p)}\right)\right)\leq s^{q-t_1\vee t_2}.
\end{align*}
Using (\ref{eqnA24}), we  arrive at
\begin{align*}
&\sum_{\substack{j_1,\ldots,j_{t_1}\text{ pw different }\\m_1,\ldots,m_{t_2}\text{ pw different }\\\left|\left\{j_1,\ldots,j_{t_1}\right\}\cap\left\{m_1,\ldots,m_{t_2}\right\}\right|=t}}\sum_{\Qpar\in\mathcal{P}\left(2M(k+1),t_1,t_1-t\right)}\sum_{\Rpar\in\mathcal{P}\left(2M(k+1),t_2,t_2-t\right)}\sum_{\substack{\ell_1^{(1)},\ldots,\ell_{k+1}^{(1)}\in T\\\vdots\\\ell_1^{(2M)},\ldots,\ell_{k+1}^{(2M)}\in T\\\ell_h^{(p)}\not=\ell_{h-1}^{p},h\in[k+1]}}
\\&\phantom{\times}\prod_{Q\in\left\{Q_{t+1},\ldots,Q_{t_1}\right\}}\delta\left(\sum_{(h,p)\in Q}(-1)^p\left(\ell_{h-1}^{(p)}-\ell_h^{(p)}\right)\right)
\\&\times\prod_{R\in\left\{R_{t+1},\ldots,R_{t_2}\right\}}\delta\left(\sum_{(h,p)\in R}(-1)^p\left(\ell_{h-1}^{(p)}-\ell_h^{(p)}\right)\right)\\
&\times\prod_{j=1}^t\delta\left(\sum_{(h,p)\in Q_j}(-1)^p\left(\ell_{h-1}^{(p)}-\ell_h^{(p)}\right)+\sum_{(h,p)\in R_j}(-1)^p\left(\ell_{h-1}^{(p)}-\ell_h^{(p)}\right)\right)\\
&\leq\sum_{\substack{j_1,\ldots,j_{t_1}\text{ pw different }\\m_1,\ldots,m_{t_2}\text{ pw different }\\\left|\left\{j_1,\ldots,j_{t_1}\right\}\cap\left\{m_1,\ldots,m_{t_2}\right\}\right|=t}}(9q^2)^q\left(\frac{3q}{2}\right)^{-t_1-t_2}s^{q-t_1\vee t_2}\\
&\leq \binom{n}{t_1}\binom{t_1}{t}\binom{n-t_1}{t_2-t}(9q^2)^q\left(\frac{3q}{2}\right)^{-t_1-t_2}s^{q-t_1\vee t_2}
\\&\leq n^{t_1}t_1^tn^{t_2-t}(9q^2)^q\left(\frac{3q}{2}\right)^{-t_1-t_2}s^{q-t_1\vee t_2}\notag\\
&\leq (9q^2s)^q\left(\frac{q}{n}\right)^t\left(\frac{n}{\frac{3}{2}q}\right)^{t_1+t_2}s^{-t_1\vee t_2},\notag
\end{align*}
where we have applied $t_1 \leq q$ in the last step. Putting these pieces together, we obtain
\begin{align}
{\cal E}&\leq\left(\frac{9q^2s}{n^2}\right)^q\sum_{t=0}^q\left(\frac{q}{n}\right)^t\sum_{t_1=t}^{q/2+\left\lfloor t/2\right\rfloor }\sum_{t_2=t}^{q/2+\left\lfloor t/2\right\rfloor }\left(\frac{n}{\frac{3}{2}q}\right)^{t_1+t_2}s^{-t_1\vee t_2}\nonumber\\
&=\left(\frac{9q^2s}{n^2}\right)^q\sum_{t=0}^q\left(\frac{q}{n}\right)^t\sum_{t_1=t}^{q/2+\left\lfloor t/2\right\rfloor }\left(\sum_{t_2=t}^{t_1-1}\left(\frac{n}{\frac{3}{2}q}\right)^{t_1+t_2}s^{-t_1}+\sum_{t_2=t_1}^{q/2+\left\lfloor t/2\right\rfloor}\left(\frac{n}{\frac{3}{2}q}\right)^{t_1+t_2}s^{-t_2}\right).\label{eqnA37}
\end{align}
Let us evaluate the inner sums in (\ref{eqnA37}). Since $n \geq (3/2) q$ by (\ref{eqnA38}) we have 
\begin{align*}
&\sum_{t_2=t_1}^{q/2+\left\lfloor t/2\right\rfloor}\left(\frac{n}{\frac{3}{2}q}\right)^{t_1+t_2}s^{-t_2}=\left(\frac{n}{\frac{3}{2}q}\right)^{t_1}\sum_{t_2=t_1}^{q/2+\left\lfloor t/2\right\rfloor}\left(\frac{n}{\frac{3}{2}qs}\right)^{t_2}\\
& = \left(\frac{n^2}{\left(\frac{3}{2}q\right)^2s}\right)^{t_1}\sum_{t_2=0}^{q/2+\left\lfloor t/2\right\rfloor-t_1}\left(\frac{n}{\frac{3}{2}qs}\right)^{t_2} 
\leq 2\left(\frac{n^2}{\left(\frac{3}{2}q\right)^2s}\right)^{q/2+t/2}.
\end{align*}
Similarly, 
using once more \eqref{eqnA38} in the form $n\geq (3/2)q\sqrt{s}$, we obtain
\begin{equation*}
\sum_{t_2=t}^{t_1-1}\left(\frac{n}{\frac{3}{2}q}\right)^{t_1+t_2}s^{-t_1}\leq\left(\frac{n^2}{\left(\frac{3}{2}q\right)^2s}\right)^{t_1}
\end{equation*}
and 
\begin{equation*}
\sum_{t_1=t}^{q/2+\left\lfloor t/2\right\rfloor}\left(\frac{n^2}{\left(\frac{3}{2}q\right)^2s}\right)^{t_1}\leq 2\left(\frac{n^2}{\left(\frac{3}{2}q\right)^2s}\right)^{q/2+t/2}.
\end{equation*}
Plugging everything into (\ref{eqnA37}) finishes the proof of the lemma.\qed
\endproof
\end{subsection}
\begin{subsection}{Proof of Theorem \ref{Thm_det_patterns}} 
{\em} According to Theorem \ref{thm_nonuniform_deterministic}, all conditions of Theorem \ref{Thm_noisy_cond} are satisfied with probability at least $1-\varepsilon$ provided 
\begin{equation*}
n^2\geq Cs\log^2\left(cN/\varepsilon\right),
\end{equation*}
where $C,c>0$ are numerical constants satisfying the bounds claimed in Theorem \ref{Thm_det_patterns}. This concludes the proof.\qed
\endproof
\end{subsection}
\end{section}

\section{Numerical simulations}
\label{sec:numerics}
\begin{subsection}{Chambolle and Pock's iterative primal dual algorithm}
For the numerical simulations, we use Chambolle and Pock's primal dual
algorithm~\cite{chpo11} to compute the solution of~\eqref{eq:l1min}
and~\eqref{eqnN2}. The algorithm is suited for a general convex optimization problem of the form 
\begin{equation}\label{FG:problem}
\min_{x\in\C^N}F(Ax)+G(x)
\end{equation}
with $A\in\C^{m\times N}$, $F:\C^m\rightarrow (-\infty,\infty]$ and $G:\C^N\rightarrow (-\infty,\infty]$ lower semi-continuous convex functions. 
The dual problem to (\ref{FG:problem}) is given by
\begin{equation}\label{FG:dual:problem}
\max_{\xi\in\C^m}-F^*(\xi)-G^*(-A^*\xi),
\end{equation}
where $F^*$, $G^*$ denote the convex conjugates of $F,G$. Here, we recall that the convex conjugate function $F^*:\C^m\rightarrow (-\infty,\infty]$ is defined as
\begin{equation*}
F^*(y):=\sup_{x\in\C^m}\left\{\Re\left(\langle x, y\rangle\right)-F(x)\right\}.
\end{equation*}
In the cases of interest to us, strong duality holds, meaning that the optimal values of (\ref{FG:problem}) and (\ref{FG:dual:problem}) coincide.
For describing Chambolle and Pock's algorithm, we require the proximal mappings of $F$ and $G$ defined as
\begin{equation*}
P_G(\tau;z):=\argmin_{x\in\C^N}\left\{\tau G(x)+\frac{1}{2}\left\|x-z\right\|_2^2\right\},
\end{equation*}
and analogously for $F$. The iterative primal dual algorithm then reads as follows. We select parameters $\theta \in [0,1]$, $\tau,\sigma > 0$ 
such that $\tau \sigma \|\textbf{A}\|_{2 \to 2} < 1$ and initial vectors $x^0 \in \C^N,\xi^0 \in \C^m$, $\bar{x}^0 = x^0$. Then one iteratively computes
\begin{align*}
\xi^{n+1} & := P_{F^*}(\sigma; \xi^{n}+\sigma A \bar{x}^{n})\;,\\
x^{n+1} & := P_{G}(\tau; x^n - \tau A^*\xi^{n+1})\;,\\
\bar{x}^{n+1} &:= x^{n+1} + \theta(x^{n+1}- x^{n})\;.
\end{align*}
In \cite{chpo11}, it is shown that for the parameter choice $\theta = 1$ the algorithm converges in the sense that 
$x^n$ converges to the minimizer of the primal problem \eqref{FG:problem}
and $\xi^n$ converges to the solution of
the dual problem \eqref{FG:dual:problem} as $n$ tends to $\infty$. Moreover, \cite{chpo11} also gives an estimate of the convergence
rate for a partial primal dual gap.


\bigskip
\end{subsection}
\begin{subsection}{The algorithm for $\ell_1$-minimization}
Let us now specialize to the case of $\ell_1$-minimization. We remark that
to the best of our knowledge, Chambolle and Pock's algorithm
has not yet been specialized to equality-constrained and noise-constrained
$\ell_1$-minimization before, so we provide the
first numerical tests of the algorithm in this setup. 

Let us first consider the problem 
\begin{equation*}
\min_{x\in\C^N}\left\|x\right\|_1 \text{ subject to } Ax =y.
\end{equation*}
This is a special case of (\ref{FG:problem}) with $G(x)=\left\|x\right\|_1$ and $F(z)=0$ if $z=y$ and $\infty$ otherwise. Straightforward computations show that for all $\xi \in \C^m$, $\zeta\in\C^N$,
\begin{align*}
F^*(\xi) & = \Re(\langle \xi,y\rangle),\qquad
G^*(\zeta)  = \chi_{B_{\|\cdot\|_\infty}}(\zeta)= \left\{ \begin{array}{ll} 0 & \mbox{ if } \|\zeta\|_\infty \leq 1,\\
\infty & \mbox{ otherwise },\end{array}\right.\\
P_F(\sigma;\xi) & = y, \qquad\qquad
P_{F^*}(\sigma;\xi)  = \xi - \sigma y.
\end{align*}
The proximal mapping of $G(x)=\left\|x\right\|_1$ can be evaluated coordinatewise, so that it is enough to compute the proximal of the modulus function $\left|\cdot\right|$ on $\C$. The latter is given by the well-known soft-thresholding operator $S_{\tau}$ defined as
\begin{align}
 S_{\tau}(z):=P_{|\cdot|}(\tau,z) 
& = \argmin_{x \in \C}\left\{ \frac{1}{2}|x-z|^2 + \tau |x| \right\}   
 = \left\{ \begin{array}{ll} \sgn(z)(|z|-\tau)& \mbox{ if } |z| \geq \tau\;,\\
 0 & \mbox{ otherwise}\;,\end{array}\right.\notag
\end{align}
so that 
\begin{equation}\label{def:soft:thresh}
P_G(\tau,z)_\ell = S_{\tau}(z_\ell), \quad \ell \in [N]\;.
\end{equation}
With these computations at hand, the algorithm for noise-free $\ell_1$-minimization is given by the iterations
\begin{align}
\xi^{n+1} & = \xi^n + \sigma (A \bar{x}^n - y)\;,\notag\\
x^{n+1} & = {\mathcal S}_{\tau}(x^n - \tau A^*\xi^{n+1})\;,\notag\\
\bar{x}^{n+1} & = x^{n+1} + \theta(x^{n+1}-x^{n})\;.\notag
\end{align}
In the noisy case, we aim at solving
\begin{equation*}
\min_{x\in\C^N}\left\|x\right\|_1 \text{ subject to } \left\|Ax-y\right\|_2\leq\eta.
\end{equation*}
In this setup, $G(x)=\left\|x\right\|_1$ and 
\begin{equation*}
F(z) = \chi_{B(y,\eta)}(z) = \left\{ \begin{array}{ll} 0 & \mbox{ if } \|z - y\|_2 \leq \eta \;,\\
\infty & \mbox{ otherwise }.\end{array} \right.
\end{equation*}
Carrying out analogous computations as in the noise-free case, we find that the corresponding algorithm for the noisy case 
consists in iteratively computing
\begin{align}
\xi^{n+1} & = \left\{ \begin{array}{l} 0 \quad \mbox{ if } \| \sigma^{-1} \xi^{n} + A \bar{x}^n - y\|_2 \leq \eta\;,\\
\big(1-\dfrac{\eta\sigma}{\|\xi^n + \sigma(A \bar{x}^n-y)\|_2}\big) (\xi^{n} + \sigma(A \bar{x}^n -y)) \;\; \mbox{ otherwise },
\end{array} \right.\notag\\
x^{n+1} & = {\mathcal S}_{\tau}(x^n - \tau A^*\xi^{n+1})\;,\notag\\
\bar{x}^{n+1} & = x^{n+1} + \theta(x^{n+1}-x^{n})\;.\notag
\end{align}
\end{subsection}
\begin{subsection}{Numerical results}
We apply the above algorithm for $\ell_1$-minimization to the sensing matrices given by (\ref{eqnRadar7}). We choose the wavelength $\lambda = 0.03\,m$, the resolution $d_0=10\,m$, the distance $z_0=10000\,m$ 
and the size of the aperture $B=30\,m$. Note that in this scenario, we have $d_0B/(\lambda z_0)=1$. 
To speed up the algorithm, we exploit the fact that the matrix 
$A\in\C^{n^2\times N}$ from (\ref{eqnRadar7}) can be factorized into a product of diagonal matrices and a nonequispaced Fourier matrix.
In fact, assuming a square resolution grid, we can write the grid parameter as double index $(\ell,\tilde{\ell})$ with $\ell, \widetilde{\ell} \in [N_1]$ where $N_1^2 = N$. 
For $j,k\in[n]$ and $a_j=(\xi_j,\eta_j)$, $a_k=(\xi_k,\eta_k)$ we then have
\begin{align*}
(Az)_{jk}&=\exp\left(\frac{\pi i}{\lambda z_0}\left(\left\|(\xi_j,\eta_j)\right\|_2^2+\left\|(\xi_k,\eta_k)\right\|_2^2\right)\right)\\
&\sum_{\ell,\tilde{\ell}\in[N_1]}\exp\left(-2\pi i\left\langle(\ell,\tilde{\ell}),\left(\frac{\xi_j+\xi_k}{B},\frac{\eta_j+\eta_k}{B}\right)\right\rangle\right)\exp\left(\frac{2\pi id_0}{\lambda z_0}\left(\ell^2+\tilde{\ell}^2\right)\right)\tilde{z}_{(\ell,\tilde{\ell})}.
\end{align*} 
Since the nonequispaced Fourier transform can be implemented at
computational costs that are only slightly larger than that of the Fast
Fourier Transform, it gives rise to fast approximate matrix-vector 
multiplication algorithms, see \cite{postta01} 
and reference therein. We use an implementation of S.~Kunis,
which can be found in the Matlab toolbox associated to the paper
\cite{kura06}. 
The algorithm is run with the renormalized matrix $\tilde{A}=\frac{1}{\sqrt{N}}A$ and the parameter choices $\theta =1$,
$\sigma=1$ and $\tau=0.5$. 
For fixed sparsity $s$, we generate a random vector in the following way: We choose the support set uniformly at random, then we sample a Steinhaus vector on this support and multiply its nonzero entries independently by a dynamic range coefficient uniformly distributed on $[1,10]$. With a fixed number of resolution cells, we vary the number $n$ of antennas and compute empirical recovery rates 
by choosing the $n$ antenna positions uniformly at random from the domain $[-B/2,B/2]^2$, where we leave the vector to be recovered fixed for the whole period. 
\begin{figure}[H]
\begin{center}
  \subfigure[]{\includegraphics[width=15cm,height=6cm]{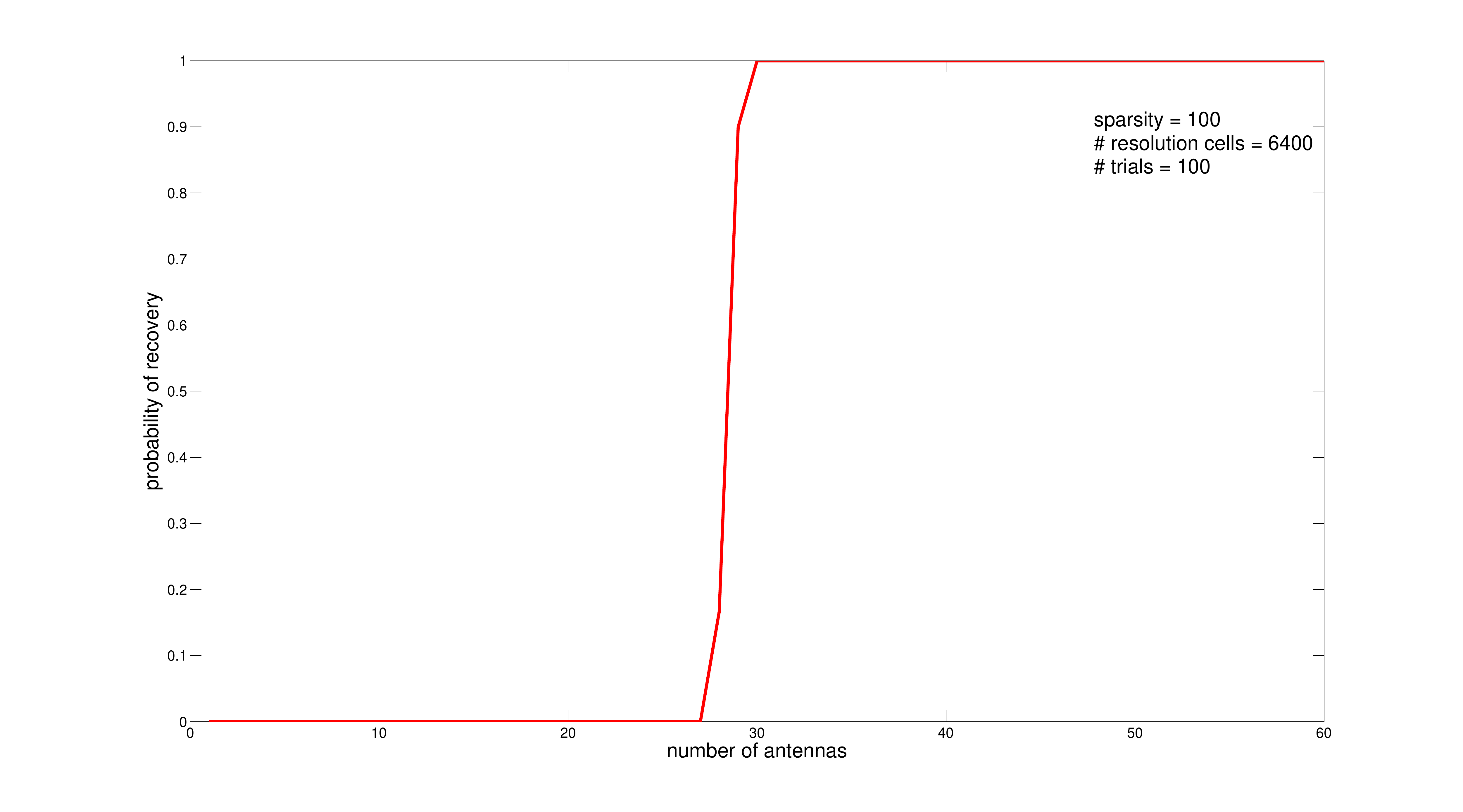}}
\subfigure[]{\includegraphics[width=15cm,height=6cm]{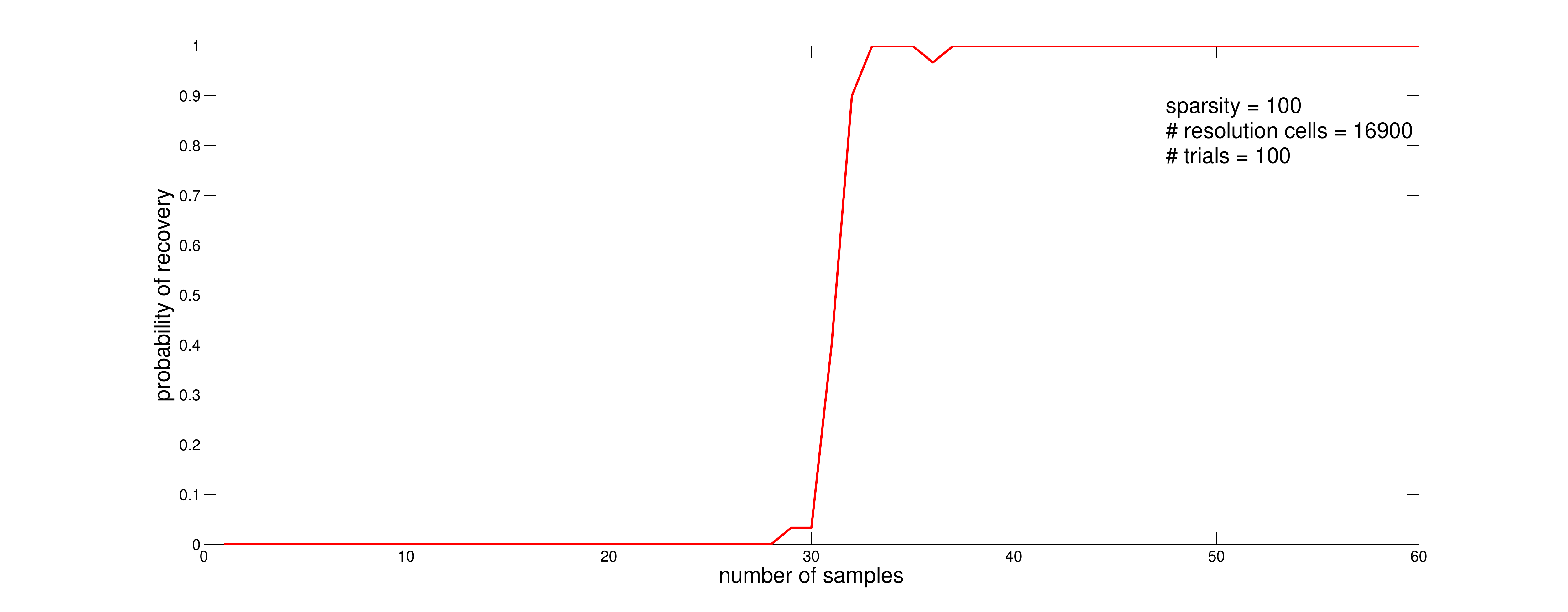}}
\end{center}
     \caption{Empirical recovery rates for fixed sparsity $s=100$ and
varying number $n$ of antennas: (a) $N=6400$ resolution cells (b)
$N=16900$ resolution cells}
     \label{figphase_transition_1}
\end{figure}
With the resulting noise-free measurement vector 
$y$ we compute the $\ell_1$-minimizer with Chambolle and Pock's algorithm (which takes about $300$ iterations), and 
we record whether the original vector is recovered (up to numerical errors of at most $10^{-3}$ measured in the $\ell_2$-norm). Repeating this test $100$ times for each choice of parameters $(s,n,N)$ provides
an empirical estimate of the success probability.
In Figure \ref{figphase_transition_1}, we display the 
result of noiseless recovery for fixed sparsity $s=100$ and for $N=6400$
respectively $N=16900$ resolution cells.
The transition from the unsuccessful regime to the successful regime occurs
at about $28$ antennas, corresponding to $784$ measurements, for $N=6400$, so in practice, the algorithm 
works even better than predicted by our theoretical results. In the situation with more resolution cells, the transition 
occurs at a slightly increased number of antennas.
The illustration in Figure \ref{figphase_transition_1} 
was produced with the version of the algorithm for equality constrained $\ell_1$-minimization.
\\To test the robustness of our recovery scheme with respect to noise, we
compute receiver operating characteristic curves for various parameter
choices, see \cite[Chapter 6]{richards05} 
and \cite[Chapter II.D]{poor94}, 
using the noise-constrained version of Chambolle and Pock's algorithm
algorithm. We start by simulating a target vector $x\in\C^{6400}$ with $\left\|x\right\|_0=100$, that is we simulate $100$ targets in $6400$ resolution cells. We do this as described above, that is we select the support uniformly at random, then simulate random phases on the support and multiply them independently by a dynamic range coefficient uniformly distributed on $[1,10]$. We then leave the vector $x$ fixed, draw a realization of our random scattering matrix $A$ and run noise constrained basis pursuit with the noisy measurements $y=Ax+e$, where $e$ is a complex Gaussian noise vector. The entries of the recovered solution $\hat{x}$  are then compared to a threshold $\tau>0$. If $\left|\hat{x}_k\right|< \tau$, then it is set to zero, otherwise it remains unchanged. We then count how many of the actual targets in $x$ are detected. The detection probability is the number of detections divided by the true number of targets, in our case $100$. Moreover, we count the number of false alarms, that is the number of positions $k\in[6400]$ where $\hat{x}_k\not= 0$ but $x_k=0$. The false alarm probability is the number of false alarms divided by the number of scatterers. For fixed $x$ and $\tau$, we repeat this a $100$ times and compute the empirical probability of detection $P_d$ and the probability of false alarm $P_f$. This is then again repeated for varying values of the threshold $\tau$, resulting in a plot of $P_d$ versus $P_f$, which is called the receiver operating characteristic curve.
\begin{figure}[H]
\centering
   \includegraphics[width=15cm,height=6cm]{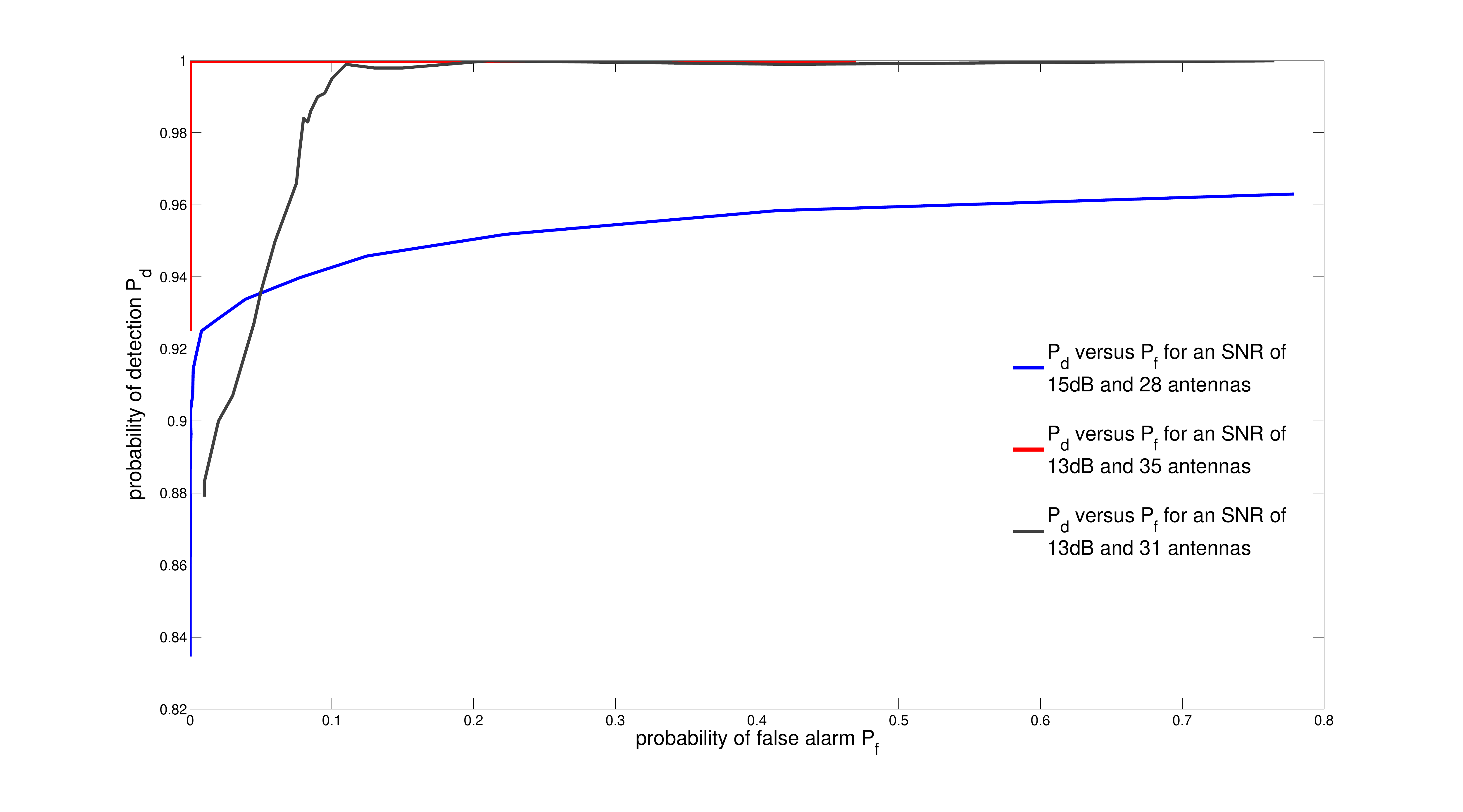}
   \caption{ROC-curves for a fixed $100$-sparse vector $x$ in $6400$ resolution cells}
\label{ROC-curve}
\end{figure} 
In Figure \ref{ROC-curve}, the results of the simulation are depicted.
We see that if we choose the number of antennas at the critical value $28$ observed in Figure \ref{figphase_transition_1}, then we get a significant number of missed targets and false alarms. If we however slightly increase the number of antennas, we get almost perfect detection and virtually no false alarms if we choose the threshold correctly, in our case as $\tau\approx 0.5$. So our recovery scheme is in fact very robust with respect to noise in the sense that the support is very well recovered. 
However, the quality of the approximation of the true reflectivities decreases with the SNR, as is to be expected.
\end{subsection}

\end{document}